\documentclass[11pt]{article}

\usepackage[margin=1in]{geometry}
\usepackage{sn-preamble}
\usepackage{latexsym}
\usepackage{amsmath}
\usepackage{amssymb}
\usepackage{amsthm}
\usepackage{qcircuit}
\usepackage{enumitem}

\newcommand{\hnote}[1]{\footnote{\textcolor{purple}{\bf Henry:} #1} }

\newcommand{\ketbra}[2]{|#1\rangle\! \langle #2|}
\newcommand{\Tr}{\mbox{\rm Tr}}

\DeclareMathOperator*{\Ex}{\mathbf{E}}

\newcommand{\RIT}{\textsc{Raw-Influence-Estimator}}
\newcommand{\QIT}{\textsc{Influence-Estimator}}
\newcommand{\QJT}{\textsc{Unitary-Junta-Tester}}
\newcommand{\TOne}{\textsc{Tester-I}}
\newcommand{\TTwo}{\textsc{Tester-II}}
\newcommand{\GGT}{\textsc{GGT}}
\newcommand{\QGGT}{\textsc{Quantum-GGT}}
\newcommand{\QLearner}{\textsc{Quantum-Junta-Learner}}

\newcommand{\Tomography}{\textsc{Tomography}}
\newcommand{\QStatePrep}{\textsc{Quantum-State-Preparation}}
\newcommand{\PauliSample}{\textsc{Pauli-Sample}}

\newcommand{\ignore}[1]{}

\title{Testing and Learning Quantum Juntas Nearly Optimally
\vspace{1em}}

\author{Thomas Chen\\[0.25em] \textsl{Columbia University} \and Shivam Nadimpalli\\[0.25em] \textsl{Columbia University}\\ \and Henry Yuen\\[0.25em] \textsl{Columbia University} \vspace{1em}}

\date{ \today }

\begin{document}

\pagenumbering{gobble}

\maketitle

\begin{abstract}
    We consider the problem of testing and learning \emph{quantum $k$-juntas}: $n$-qubit unitary matrices which act non-trivially on just $k$ of the $n$ qubits and as the identity on the rest. As our main algorithmic results, we give 
    \begin{enumerate}
        \item A $\widetilde{O}(\sqrt{k})$-query quantum algorithm that can distinguish quantum $k$-juntas from unitary matrices that are ``far'' from every quantum $k$-junta; and
        \item A $O(4^k)$-query algorithm to \emph{learn} quantum $k$-juntas.
    \end{enumerate}
    
    We complement our upper bounds for testing and learning quantum $k$-juntas with near-matching lower bounds of $\Omega(\sqrt{k})$ and $\Omega(4^k/k)$, respectively. Our techniques are Fourier-analytic and make use of a notion of \emph{influence} of qubits on unitaries.%, first introduced by Montanaro~and~Osborne \cite{Montanaro2010} for Hermitian unitaries.
\end{abstract}

\newpage
\pagenumbering{arabic}
\section{Introduction}
\label{sec:intro}

Certifying and characterizing the dynamical behavior of quantum systems is a fundamental task in physics which is often achieved via \emph{quantum process tomography} (QPT) \cite{chuang1997prescription}. However, QPT is extremely resource-intensive. For example, all known methods for learning a classical description of an arbitrary $n$-qubit unitary operator, given black-box query access to it, require $\Omega(4^n)$ queries to the unitary~\cite{gutoski2014process}. On the other hand, this complexity can be significantly reduced if, instead of learning the entire description of the unknown unitary, we want to \emph{test} whether the unitary satisfies a specific property. This naturally leads us to consider the well-studied \emph{property testing} framework in theoretical computer science \cite{Goldreich2010,Bhattacharyya2022}.

The setup of property testing (in the context of unitary dynamics, as it pertains to this paper) is as follows: Given {oracle access}\footnote{More formally, an oracle for a unitary $U$ takes in as input a quantum state $\ket{\psi}$ and outputs $U\ket{\psi}$.} to a unitary operator $U$ and its inverse $U^\dagger$, our goal is to determine whether $U$ has a certain property or is ``far''\footnote{See \Cref{def:distance} for a formal definition of ``$\dist$,'' the distance metric.} from every unitary operator satisfying that property using a small number of calls to the oracles to $U$ and $U^\dagger$. We also allow for the algorithm to output an incorrect answer with some small probability. Several natural properties of unitary dynamics have been studied in this model, such as commutativity, diagonality, membership in the Pauli basis, etc. We refer the interested reader to Section~5.1 of the survey by Montanaro~and~de~Wolf on quantum property testing \cite{Montanaro2016} for more information. 

The property we are interested in testing here is that of being a \emph{$k$-junta}: We say that an $n$-qubit unitary $U$ is a \emph{$k$-junta} if it acts ``non-trivially'' on only $k$ of the $n$-qubits (see \Cref{def:quantum-junta} for a formal definition). Like Montanaro~and~Osborne~\cite{Montanaro2010}, we will refer to a unitary $k$-junta as a \emph{quantum $k$-junta}, to distinguish it from a $k$-junta Boolean function (or simply, Boolean $k$-junta). As a special case, the notion of quantum $k$-juntas captures the well-studied problem of testing if a Boolean function $f\isazofunc$ is a $k$-junta (cf. \Cref{question:bjt}).%; such functions arise naturally in learning theory when most of the features are irrelevant to the concept to be learned. 

\begin{problem}[Testing quantum $k$-juntas] \label{question:qjt}
	Given oracle access to a unitary $U$ and its inverse $U^\dagger$ acting on $n$ qubits and $\eps > 0$, decide with probability at least $9/10$ if $U$ is a $k$-junta or if $\dist(U, V) \geq\eps$ for all quantum $k$-juntas $V$ acting on $n$ qubits. 
\end{problem}

Our first main result is an algorithm for testing if a unitary $U$ is $k$-junta using $\wt{O}(\sqrt{k})$ queries to $U$ and $U^\dagger$, where $\wt{O}(\cdot)$ hides polylogarithmic factors of $k$. Crucially, the query complexity of the tester is independent of $n$, the total number of qubits in $U$. We complement this with a near-matching lower bound of $\Omega(\sqrt{k})$ for the junta testing problem, implying that our algorithm is optimal up to a polylogarithmic factor in $k$.

\begin{theorem}[Informal version of \Cref{thm:main-tester,thm:reduction}]
	Quantum $k$-juntas can be tested with $\wt{O}    (\sqrt{k})$ queries. Furthermore, testing quantum $k$-juntas requires $\Omega(\sqrt{k})$ queries.
\end{theorem}

Our second main result is an algorithm to \emph{learn} quantum $k$-juntas with significantly lower sample complexity than the naive QPT approach; in particular there is \emph{no} dependence on the total number of qubits, $n$. 
%sample complexity independent of $n$, drastically improving upon the naive approach via QPT; in particular, we have the following guarantee. 

\begin{theorem}[Informal version of \Cref{thm:learning,thm:learning-lb}]
	Given oracle access to a quantum $k$-junta $U$ acting on $n$-qubits and $\eps > 0$, there exists an algorithm that makes $O(4^k/\eps^2)$ queries to $U$ and outputs with probability $9/10$ a unitary $\wh{U}$ such that $\dist(U,\wh{U})\leq \eps$. Furthermore, $\Omega(4^k/k)$ queries are necessary to learn quantum juntas.
\end{theorem}

Both our upper bounds (for testing and for learning) are proved via Fourier-analytic techniques and crucially make use of the notion of \emph{influence of qubits on a unitary}, first introduced by Montanaro~and~Osborne~\cite{Montanaro2010} in the context of Hermitian unitary matrices. Our lower bound for testing quantum $k$-juntas appeals to the lower bound for testing Boolean $k$-juntas obtained by Bun, Kothari, and Thaler~\cite{Bun2017}, as well as a new structural result for quantum $k$-juntas. Our lower bound for learning quantum $k$-juntas arises from the communication complexity of the \textsc{Input Guessing} game \cite{nayak1999optimal}.
% We give a technical overview of our results in \Cref{subsec:our-stuff}.

\paragraph{Organization.} We briefly recall related work on testing both Boolean and quantum juntas in \Cref{subsec:related-work}, and then give a high-level technical overview of our results in \Cref{subsec:our-stuff}. We prove our $\wt{O}(\sqrt{k})$ upper bound for testing quantum $k$-juntas in \Cref{sec:testing-ub}, and prove our $\Omega(\sqrt{k})$ lower bound for the same in \Cref{sec:lower-bound}. Finally, we present our upper and lower bound on learning quantum $k$-juntas in \Cref{sec:learning}.             

\subsection{Related Work}
\label{subsec:related-work}
We summarize related work as well as our contributions in \Cref{table:prior-work}.

{\renewcommand{\arraystretch}{1.5}
\begin{table}[]
\centering
\begin{tabular}[]{@{}llll@{}}
\toprule
                  & Classical Testing                                                  & Quantum Testing & Quantum Learning \\ \midrule
$f\isazofunc$ & \begin{tabular}[t]{@{}l@{}}${O}(k\log k)$ \cite{Blais2009} \\ $\Omega(k\log k)$ \cite{Saglam2018}\end{tabular} &   \begin{tabular}[t]{@{}l@{}}$\wt{O}(\sqrt{k})$ \cite{Ambainis2016} \\ $\Omega(\sqrt{k})$ \cite{Bun2017} \end{tabular}              &
\begin{tabular}[t]{@{}l@{}}$O(2^k)$ \cite{Atici2007} \\ $\Omega(2^k)$ \cite{Atici2007} \end{tabular}                  \\ 
Unitary $U \in \calM_{2^n\times2^n}$ &                                                                    --- &    \begin{tabular}[t]{@{}l@{}}$\wt{O}(\sqrt{k})$ (\Cref{thm:main-tester}) \\ $\Omega(\sqrt{k})$ (\Cref{thm:reduction}) \end{tabular}             &
\begin{tabular}[t]{@{}l@{}}$O(4^k)$ (\Cref{thm:learning})  \\ $\Omega(\frac{4^k}{k})$ (\Cref{thm:learning-lb}) \end{tabular}             \\ \bottomrule
\end{tabular}

\caption{Our contributions and prior work on testing and learning Boolean and quantum $k$-juntas.} \label{table:prior-work} 
\end{table}
}

\paragraph{Classical Testing of Boolean Juntas.} We first note that \Cref{question:qjt} captures as a special case its Boolean analog, which we state below as \Cref{question:bjt}. Recall that a Boolean function $f\isazofunc$ is a $k$-junta if $f(x) = g(x_{i_1}, \ldots, x_{i_k})$ for some $g:\zo^k \to \zo$. We also say that for $f, g\isazofunc$, 
\[\dist(f,g) := \Pr[f(\bx)\neq g(\bx)]\] for $\bx\sim\zo^n$ 
drawn uniformly at random. (In other words, the distance metric we use for Boolean functions is simply the normalized Hamming distance.)

\begin{problem}[Testing Boolean $k$-juntas]
\label{question:bjt}
	Given classical or quantum query access to a function $f\isazofunc$ via a unitary $\calO_f$, decide with constant probability if $f$ is a $k$-junta or if $\dist(f,g)\geq \eps$ for every $k$-junta $g\isazofunc$. 
\end{problem} 

This question has been extensively studied over recent decades, with the first result explicitly related to testing juntas obtained by Parnas, Ron, and Samorodnitsky \cite{Parnas2002} who gave a \emph{classical} algorithm for testing $1$-juntas with $O(1)$ queries. Soon afterwards, Fischer et al. \cite{Fischer2004} introduced classical algorithms for testing $k$-juntas with $\wt{O}(k^2)$ queries. The query complexity of classically testing juntas was later improved by Blais \cite{Blais2009} who gave a nearly optimal tester which makes $\wt{O}(k)$ queries. Blais's tester is asymptotically optimal up to a logarithmic factor, given the $\Omega(k)$ lower bound for classically testing $k$-juntas by \cite{Chockler2004}.

\paragraph{Quantum Testing and Learning of Boolean Juntas.} There has also been a long line of work on testing Boolean juntas via \emph{quantum} algorithms, i.e. algorithms with query access to a unitary $\calO_f$ representing a function $f\isazofunc$, allowing the algorithm to query superpositions of inputs. At\i c\i\ and Servedio \cite{Atici2007} gave an elegant quantum algorithm to test $k$-juntas using $O(k)$ queries.\footnote{At the time \cite{Atici2007} was written, the best classical upper bound on testing juntas was $\wt{O}(k^2)$.} More recently, Ambainis et al. \cite{Ambainis2016} came up with a quantum algorithm to test juntas that makes only $\wt{O}(\sqrt{k})$ queries. This was shown to be essentially optimal by Bun, Kothari, and Thaler~\cite{Bun2017} who proved an $\wt{\Omega}(\sqrt{k})$ lower bound for via a reduction from the image size testing problem. Finally, At\i c\i\ and Servedio \cite{Atici2007} also gave a $O(2^k)$-sample quantum algorithm for \emph{learning} Boolean $k$-juntas in the PAC model. 

\paragraph{Quantum Testing of Quantum Juntas.} Returning to \Cref{question:qjt}, Wang \cite{Wang2011} gave a tester for testing whether a unitary operator $U$ is a $k$-junta or is $\eps$-far from a $k$-junta that makes $O(k)$ queries, and their algorithm turns out to be a direct generalization of the tester of At\i c\i\ and Servedio \cite{Atici2007}.\footnote{The result originally obtained by Wang had a worse bound of $O(k\log k)$, but this can be improved to $O(k)$ by following the analysis of \cite{Atici2007} (cf. Section 5.1.6 of \cite{Montanaro2016}). We also note that their query complexity's dependence on $\eps$ can be improved via a straightforward application of amplitude amplifiaction.} Finally, Montanaro and Osborne \cite{Montanaro2010} had previously studied a different tester for the property of being a ``dictatorship,'' i.e. a $1$-junta, but did not prove correctness. 

\subsection{Our Techniques}
\label{subsec:our-stuff}

In this section, we give a high-level technical overview of our main results. 

\subsubsection{Testing Quantum Juntas} 

Our $\wt{O}(\sqrt{k})$-query tester for quantum $k$-juntas can be viewed as direct analog of the $\wt{O}(\sqrt{k})$-query tester for Boolean $k$-juntas obtained by Ambainis, et al. \cite{Ambainis2016}. Our tester relies crucially on the notion of \emph{influence of qubits on a unitary}, which was first introduced by Montanaro~and~Osborne~\cite{Montanaro2010} for Hermitian unitaries. Informally, the influence of a qubit on a unitary $U$ captures how non-trivially $U$ acts on that qubit; see \Cref{subsec:prelims-influence} for a formal definition as well as useful properties of this notion of influence. Our main technical contributions here are an alternate formulation of the influence of a qubit on a unitary and a subroutine \QIT\ (cf. \Cref{subsec:testable-influence}) to estimate this influence. With this in hand, we closely mirror the approach of Ambainis et al. \cite{Ambainis2016} in \Cref{subsec:ambainis}. We essentially used their algorithm as a black-box, but our analysis differs in certain parameters; for completeness, we present the entire analysis with these modifications.

The $\Omega(\sqrt{k})$-query lower bound for testing quantum $k$-juntas relies on the $\Omega(\sqrt{k})$-query lower bound for testing {Boolean} $k$-juntas obtained by Bun, Kothari, and Thaler~\cite{Bun2017}. We do so via the natural encoding of a Boolean function $f\isazofunc$ as a unitary $U_f$ given by 
\begin{equation}
    \label{eq:encoding}
    U_f := \mathrm{diag}((-1)^{f(x)}).
\end{equation}
It is immediate from \Cref{eq:encoding} that encoding a Boolean $k$-junta in this way yields a quantum $k$-junta. Our main structural result \Cref{prop:lb-no-case}, shows that if a Boolean function $f\isazofunc$ is far from any Boolean $k$-junta, then $U_f$ is also far from any quantum $k$-junta. We start by first showing that $U_f$ is far from $U_g$ for every Boolean $k$-junta $g\isazofunc$, and then handling  arbitrary quantum $k$-juntas via \Cref{lemma:lb-structural-result}.

\subsubsection{Learning Quantum Juntas} 

Our learning algorithm, \Cref{alg:q-learner}, can be viewed as analogous to the algorithm obtained by At\i c\i~and~Servedio \cite{Atici2007} for Boolean $k$-juntas. We make use of the \emph{Choi-Jamio\l kowski (CJ) isomorphism} between unitary operators on $n$ qubits and pure states in $\C^{2^n}\times\C^{2^n}$, allowing us to techniques used to learn \emph{quantum states} to learn the unitary $U$. We start by first determining the high-influence qubits of the quantum $k$-junta $U$ via ``Pauli sampling,'' which is analogous to the Fourier Sampling subroutine used by \cite{Atici2007}. We then take the CJ isomorphism of $U$ and appropriately trace out the qubits with negligible influence, allowing us to use a \emph{pure state tomography} procedure to learn the reduced CJ state using $O(4^n/\eps^2)$ samples.

Our lower bound for learning quantum juntas is proved via a reduction to a communication complexity lower bound, namely a quantum lower bound proved by Nayak~\cite{nayak1999optimal} on the communication required for one party to guess the input of another party.

\subsection{Future Work}
\label{subsec:future}

A natural next direction is to consider the testability/learnability of \emph{quantum channels} acting non-trivially on $k$-qubits. Recall that a quantum channel is a completely positive, trace-preserving linear map; see \cite{Watrous2018} for a comprehensive introduction to the subject. As noted in \cite{Montanaro2016}, there has not been much work on testing properties of quantum channels.

We also remark that (to our knowledge) there has been no work on \emph{tolerant} property testing---for both Boolean functions as well as unitary matrices---via quantum algorithms.\footnote{Recall that in the tolerant model, the tester is asked to distinguish instances that are $\eps_1$-close to the property from instances that are $\eps_2$-far from the property.} The best known classical upper bound for tolerant testing of Boolean $k$-juntas is $2^{\wt{O}(\sqrt{k})}$ due to Iyer, Tal, and Whitmeyer~\cite{Iyer2021}. We also note that a $\Omega(2^{\sqrt{k}})$ lower-bound against classical non-adaptive algorithms for tolerant junta testing was obtained by Pallavoor~et al.~\cite{PRW-mono}. 

Finally, it is unknown whether quantum algorithms offer any advantage in terms of query complexity for the problem of testing Boolean $k$-juntas in the \emph{distribution-free} setting.\footnote{In the \emph{distribution-free} model, the distance between two functions is measured with respect to a fixed but unknown distribution.} In particular, Belovs~\cite{Belovs2019} gave a $O(k)$ quantum tester for Boolean $k$-juntas in the distribution-free model, matching the query complexity of the best classical algorithms for testing Boolean $k$-juntas in the distribution-free model due to Bshouty~\cite{Bshouty2019} and Zhang~\cite{Zhang2019}.

\section{Preliminaries}
\label{sec:prelims}

In this section, we introduce notation and recall useful background. We assume familiarity with elementary quantum computing and quantum information theory, and refer the interested reader to \cite{Nielsen2010,Wilde2017} for more background. For $n \geq 1$, we will write $N = 2^n$. Given $T\sse[n]$, we will write $\overline{T} := [n] \setminus T$. We will write $I_n$ to denote the $n\times n$ identity matrix; when $n$ is clear from context, we may write $I$ instead. 

\subsection{Unitary Operators}
\label{subsec:unitary}

We will write $\calM_{N,N}$ to denote the set of linear operators from $\C^N$ to $\C^N$ and denote by $\calU_{N}$ the set of $N$-dimensional unitary operators, i.e. 
\[\calU_{N} := \cbra{ U \in \calM_{N,N} : UU^\dagger = U^\dagger U = I}.\]

\begin{definition} \label{def:partial-trace}
Given a unitary $U\in\calU_N$ and $S\sse[n]$, we define the operator $\Tr_S(U)$ obtained by \emph{tracing out $S$} to be 
\[\Tr_S(U) = \sum_{k \in \{0, 1\}^S} (I_{\overline{S}} \otimes \bra{k}) U (I_{\overline{S}} \otimes \ket{k}).\]
\end{definition}

In the above definition, we write $\ket{k}$ for $k\in \zo^S$ to be the $|S|$ qubit state in the computational basis corresponding to the bit-string $k$. Note that \Cref{def:partial-trace} aligns with the fact that the trace of a unitary matrix $U$ is given by 
\[\Tr(U) = \sum_{k \in \{0, 1\}^n} \bra{k} U  \ket{k}.\]

\begin{definition}[$k$-Junta] \label{def:quantum-junta}
	We say that a unitary $U\in\calU_N$ is a \emph{quantum $k$-junta} if there exists $S \sse [n]$ with $|S| = k$ such that 
	\[U = V_S \otimes I_{{\overline{S}}}\]
	for some $V_S \in \calU_{2^{k}}$.
\end{definition}

In contrast, \emph{classical} $k$-juntas are Boolean functions $f\isazofunc$ that depend on only $k$ of their $n$ input variables. More formally, a function $f\isazofunc$ is a $k$-junta if there exists $g:\zo^k\to\zo$ such that $f(x_1, \ldots, x_n) = g(x_{i_1}, \ldots, x_{i_k})$ for some fixed $i_1, \ldots, i_n \in [n]$ and for all $x\in\zo^n$.

We will view $\calM_{N,N}$ as an inner-product space equipped with the Hilbert--Schmidt inner product 
\[\abra{A,B} := \Tr\pbra{A^\dagger B}.\]
Recall that the Hilbert--Schmidt inner product induces the Hilbert--Schmidt (or Frobenius) norm, which is given by 
\[\|A\|^2 := \Tr\pbra{A^\dagger A} = \sum_{i,j = 0}^{N - 1} |A[i,j]|^2.\]

We will use the following metric to compare the distance between unitary matrices. {Note that this metric is \emph{not} the natural metric induced by the Hilbert--Schmidt norm. }
\begin{definition}[Distance between unitaries] \label{def:distance}
	Given $A, B \in \calM_{N,N}$, we define 
	\[\dist(A,B) := \min_{\theta\in[0,2\pi)} \frac{1}{\sqrt{2N}} \|e^{i\theta}A - B\|.\]
	We say that $A$ is \emph{$\eps$-far} from $B$ if $\dist(A,B) \geq \eps$. More generally, for any $\calP \subseteq \calM_{N,N}$ and $A\in\calM_{N,N}$, we write 
	\[\dist(A,\calP) := \min_{B \in \calP} \dist(A,B)\]
	and similarly say that $A$ is \emph{$\eps$-far} from $\calP$ if $\dist(A,\calP)\geq\eps$.
\end{definition}

It can easily be checked that $\dist(A,B) \geq 0$, with equality holding if and only if $A = e^{i\theta}B$ for some $\theta\in[0,2\pi)$, as well as other standard properties of a metric. Finally, note that $\dist(V_1 \otimes U, V_2 \otimes U) = \dist(V_1, V_2)$ for unitaries $U, V_1, V_2$.

\subsection{The Pauli Decomposition}
\label{subsec:pauli}

In this section, we introduce a useful orthonormal basis for $\calM_{N,N}$ (viewed as a $\C$-vector space) which will be central to what follows. Recall that the set of Pauli operators given by 
\[
\sigma_0 = \begin{pmatrix}
	1 & 0 \\ 0 & 1 
\end{pmatrix} = I,\quad
\sigma_1 = \begin{pmatrix}
	0 & 1 \\ 1 & 0 
\end{pmatrix} = X,\quad
\sigma_2 = \begin{pmatrix}
	0 & -i \\ i & 0
\end{pmatrix} = Y,\quad\text{and}\quad
\sigma_3 = \begin{pmatrix}
	1 & 0 \\ 0 & -1 
\end{pmatrix} = Z
\]
forms an orthonormal basis for $\calM_{2,2}$ with respect to the Hilbert--Schmidt inner product. For $x \in \cbra{0,1,2,3}^n \cong \Z_4^n$, we define $\sigma_x := \sigma_{x_1}\otimes\cdots\otimes\sigma_{x_n}$ and write $\supp(x) := \{i\in[n] : x_i \neq 0\}$. It is then easy to check that the collection 
\[\cbra{\frac{1}{\sqrt{N}} \sigma_x }_{x\in\Z_4^n}\] forms an orthonormal basis for $\calM_{N,N}$ with respect to the Hilbert--Schmidt inner product. We will frequently refer to this basis as the \emph{Pauli basis} for $\calM_{N,N}$. It follows that we can write any $A\in\calM_{N,N}$ as 
\[A = \sum_{x \in \Z_4^n} \wh{A}(x)\sigma_x \qquad\text{where}\qquad \wh{A}(x) := \frac{1}{N}\abra{A,\sigma_x}.\]
We will sometimes refer to $\wh{A}(x)$ as the \emph{Pauli coefficient of $A$ on $x$} and will refer to the collection $\{\wh{A}(x)\}_x$ as the \emph{Pauli spectrum of $A$.} It is easy to verify that Parseval's and Plancharel's formulas hold in this setting: 
\[\frac{1}{N} \|A\|^2 = \sum_{x\in\Z_4^n} |\wh{A}(x)|^2 \qquad\text{and}\qquad \frac{1}{N}\abra{A,B} = \sum_{x\in\Z_4^n} {\wh{A}(x)}^\ast\cdot\wh{B}(x).\]
In particular, for $U\in\calU_N$, we have $\sum_{x\in\Z_4^n} |\wh{U}(x)|^2 = 1$. 

\subsection{Influence of Qubits on Unitaries}
\label{subsec:prelims-influence}

\cite{Montanaro2010} introduced a notion of \emph{influence of qubits on unitaries}, in the spirit of the well-studied classical notion of influence of variables on Boolean functions $f: \zo^n \to \zo$ (cf. Chapter~2~of~\cite{ODonnell2014a}). This notion of influence will be central to the testing algorithm presented in \Cref{sec:testing-ub}. Although \cite{Montanaro2010}'s notion of influence was developed only for Hermitian unitaries (i.e. ``Quantum Boolean Functions''), we first present their formulation as it gives good intuition for what influence captures, after which we introduce a more general definition of influence that applies to arbitrary unitaries as well as to more than one qubit.

\begin{definition}[Derivative operator] \label{def:derivative}
	The $i^\text{th}$ \emph{derivative operator} $\D_i$ is a superoperator on $\calM_{N,N}$ defined through its action on the Pauli basis element $\sigma_x,$ $x\in\Z_4^n$: 
	\[\D_i \sigma_x = \begin{cases}
 \sigma_x & x_i \neq 0\\ 0 & x_i = 0	
 \end{cases}.
\]
	%For $S\sse[n]$ where $S = \{i_1, \ldots, i_k\}$, we write $\D_S$ to mean 
	%\[\D_S := \D_{i_1}\circ\ldots\circ\D_{i_{k}}\]
\end{definition}

It follows immediately that for $A\in\calM_{N,N}$, $A = \sum_{x \in \Z_4^n} \wh{A}(x) \sigma_x$, we have 
\begin{equation} \label{eq:fourier-def-derivative}
	\D_iA = \sum_{x : x_i \neq 0} \wh{A}(x)\sigma_x %\qquad\text{and}\qquad {\D_SA = \sum_{x : \supp(x) \cap S \neq \emptyset}\wh{A}(x)\sigma_x.}
\end{equation}
Informally, $\D_i$ isolates the part of the Pauli spectrum that acts non-trivially on the $i^{\text{th}}$ qubit (i.e. the $x$ such that $\sigma_{x_i} \neq I$). %this generalizes naturally to $\D_S$. 
We can now introduce the notion of {influence} of qubits on unitaries proposed by \cite{Montanaro2010}. %\hnote{Dumb question; if I think of $D_{i_j}$ as killing off any $\sigma_x$ such that $x_{i_j} = 0$, then wouldn't $D_S$, which is the sequential composition of the $D_{i_j}$'s for $i_j \in S$, kill off all $\sigma_x$ unless $x_S = 1$? In other words, the summation should be over $x$ such that $S \subseteq \supp(x)$?}

\begin{definition}[Influence of single qubit] \label{def:influence}
	Given a unitary $U\in\calU_N$, the \emph{influence of the $i^\text{th}$ qubit on $U$}, written $\Inf_i[U]$, is 
	\[\Inf_i[U] := \|\D_iU\|^2.\]
	%More generally, for $S \subseteq [n]$ we define the \emph{influence of $S$ on $U$} $\Inf_S[U]$ as 
	%\[\Inf_S[U] := \|\D_SU\|^2.\]
\end{definition}

At a high level, the influence of the $i^\text{th}$ qubit on a unitary $U$ captures how non-trivially the unitary $U$ acts on the $i^\text{th}$ qubit of a quantum state. Note that it is immediate from \Cref{eq:fourier-def-derivative} that 
\[\Inf_i[U] = \sum_{x : x_i\neq 0} |\wh{U}(x)|^2.\]
This suggests a natural way to extend the \Cref{def:influence} to more than one qubit. 

\begin{definition}[Influence of multiple qubits]
\label{def:influence-multiple-qubits}
	Given a unitary $U\in\calU_N$ and $S\sse[n]$, the \emph{influence of $S$ on $U$}, written $\Inf_S[U]$, is 
	\begin{align} \label{eq:analytic-inf}
	\Inf_S[U] &= \sum_{x : \supp(x)\cap S\neq\emptyset} |\wh{U}(x)|^2.
	\end{align}
\end{definition}

%The latter definition generalizes \cite{Montanaro2010}'s analytic definition of influence for a single qubit to a subset of qubits in the natural way: First, it is 
The above definition is analogous to the ``Fourier formula'' for the the influence of a set of variables on a Boolean function (cf. Section~2.4 of~\cite{Ambainis2016}). 
%Furthermore, this definition enables us to write Wang's crucial \Cref{lemma:eps-far-guarantees}~\cite{Wang2011} in terms of this notion of influence. 
Furthermore, as stated earlier, note that these definitions apply to arbitrary unitaries (i.e. we do not require them to be Hermitian). We present an alternative characterization of $\Inf_S[U]$ (which we will not require, but may be of independent interest) in \Cref{appendix:influence}. We have the following lemma.  

%\begin{remark}
	%Note that \cite{Montanaro2010} originally introduced \Cref{def:derivative,def:influence} for Hermitian unitaries $U$, and for a single qubit $j$; here, we extend these definitions naturally to all unitaries. \hnote{and for more qubits?} 
%\end{remark}
 
%Here, as all unitaries have a Pauli decomposition, we are extending this notion to all unitaries, and for sets of qubits $S$ rather than a single qubit $j$. Roughly, influence corresponds to how much U acts non-trivially on qubit j of an input state $\ket{\psi} \in C^{2^n}$, with lower influence indicating that it acts more like the identity on that qubit. We will now prove some useful properties of this general influence.

%\subsection{Key Properties of Influence}
%\label{subsec:prelims-influence-properties}

\begin{lemma} \label{lemma:inf-useful-props}
	For $S, T \sse [n]$ and a unitary $U \in \calU_N$, we have 
	\begin{enumerate}
		\item Monotonicity: If $S\sse T$, $\Inf_S[U] \leq \Inf_T[U]$; and 
		\item Subadditivity:  $\Inf_{S \cup T}[U] \leq \Inf_S[U] + \Inf_T[U]$.
	\end{enumerate}
\end{lemma}

Note that monotonicity is immediate from the analytic interpretation of influence (cf. \Cref{eq:analytic-inf}), and subadditivity follows from the fact that 
\[\{(S \cup T) \cap \supp(x)\} = \{S \cap \supp(x)\} \cup \{T \cap \supp(x)\}.\]

As mentioned before, Wang \cite{Wang2011} implicitly used this notion of influence to test quantum $k$-juntas. In particular, Wang proved the following.

\begin{lemma}[\cite{Wang2011}]\label{lemma:eps-far-guarantees}
Given a unitary $U\in\calU_N$, if $U$ is $\eps$-far from every quantum $k$-junta $V$, then for all $T\sse[n]$ with $|T|\leq k$, we have that 
\[\Inf_{\overline{T}}[U] \geq \frac{\eps^2}{4}.\]
%If $D(U, \widetilde{V} \otimes I_{T^c}) > \eps$ for all $T \subset [n], |T| \leq k$ and for all $\widetilde{V} \in U_{2^{|T|}}$, then for all $T \subset [n], |T| \leq k$, we have $\Inf_{T^c}[U] \geq \eps^2/4$
\end{lemma}

\subsection{Query Complexity of the Composition of Quantum Algorithms}\label{sec:irrel-var}

Our upper bound for quantum junta testing will compose quantum algorithms, each having some probability of error. We will invoke the following lemma about the query complexity of composed algorithms.

\begin{lemma} [\cite{Ambainis2016}, Corollary 2.12]\label{lemma:fn-composition}
With $D \subset \{ 0, 1\}^n$, let $F : D \rightarrow \{ 0, 1\}$ and $G_j$ be partial Boolean functions $ \forall j \in [n]$. Let $Q(F)$ denote the bounded-error quantum query complexity of $F$. Let $T$ equal to the objective value of a feasible solution $(X_j)$ to the adversarial bound in (2.3) of \cite{Ambainis2016}. We let an input variable $j$ be irrelevant for input $z \in D$ if and only if $X_j[z,z] = 0$. Then, we have
\begin{align*}
Q(F \circ (G_1, ... G_n)) = O\left(T \max_{j \in [n]} Q(G_j)\right).
\end{align*}
with the function composition done as in Definition 2.10 of \cite{Ambainis2016}.

\end{lemma}

First, the result says that we can compose functions without any overhead. Ordinarily, if we compose two algorithms that have some probability of error, then we may need to amplify the success probability of each subroutine, inducing a logarithmic overhead in query complexity.

A second important aspect of the result is that functions can be composed in a way where the top-level function, $F$, ignores ``irrelevant" inputs. In the context of our setting, $F$ is a Group-Testing algorithm, while each $G_j$ is an instance of an influence tester, applied on a particular subset of qubits. The formalism of \textit{irrelevant variables} allows the group tester to ignore certain instances of influence testers whose behavior is unpredictable. The details of our approach is identical to that of \cite{Ambainis2016}, to which we defer the technical details.

\subsection{The Choi-Jamiolkowski Isomorphism}\label{sec:cj}

In our algorithms, we will encode a unitary as a quantum state using the \emph{Choi-Jamio\l kowski isomorphism} \cite{Choi1975,Jamiolkowski1972}, which is a mapping between $N \times N$ unitary operators and pure states in $\C^N \otimes \C^N$. Concretely, this mapping associates to every unitary $U \in \calU_N$ the \emph{Choi-Jamio\l kowski state} (which we abbreviate as \emph{CJ state}):
  \[\ket{v(U)} := (U \otimes I) \pbra{\frac{1}{\sqrt{N}} \sum_{0 \leq i < N} \ket{i}\ket{i}}  = \frac{1}{\sqrt{N}} \sum_{0 \leq i,j < N} U[i,j] \, \ket{i} \ket{j}.\]
The CJ state $\ket{v(U)}$ can be prepared by first creating the maximally entangled state of dimension $N$, and then querying $U$ on half of the maximally entangled state. Since $N = 2^n$, this is equivalent to preparing $n$ EPR pairs (which altogether forms $2n$ qubits) and applying the unitary $U$ to the $n$ qubits coming from the first half of each of the EPR pairs. As such, each qubit of the unitary $U$ corresponds to two qubits of the state $\ket{v(U)}$. We will refer to qubits in $\{ 1, \ldots, n\}$ as the ones acted on by the unitary $U$, and qubits in $\{n + 1,\ldots,\ 2n\}$ as the ones acted on by $I$. We introduce the following notation for convenience.

\begin{notation}
For each qubit $\ell\in [n]$ acted on by the unitary $U$, there is a pair of corresponding qubits $(\ell, \wt{\ell})\in [n] \times \{n+1,\ldots,2n\}$ in the state $\ket{v(U)}$. In particular, $\wt{\ell}$ and $\ell$ are related as they formed an EPR pair at the synthesis of the CJ state.
\end{notation}

\section{Testing Quantum $k$-Juntas with $\wt{O}(\sqrt{k})$ Queries}
\label{sec:testing-ub}

As suggested by \Cref{lemma:inf-useful-props}, the notion of influence for unitaries behaves analogously to the ``usual'' notion of influence for Boolean functions, which was crucial to the $\wt{O}(\sqrt{k})$-query $k$-junta tester for Boolean functions obtained by Ambainis et al. \cite{Ambainis2016}. This motivates an analog of the algorithm obtained by Ambainis et al. for quantum juntas, and this is indeed how we obtain a $\wt{O}(\sqrt{k})$-tester for quantum $k$-juntas. In \Cref{subsec:testable-influence}, we present an unbiased estimator for the influence of qubits on a unitary, which we then combine with Ambainis et al.'s tester in \Cref{subsec:ambainis} to obtain our quantum $k$-junta tester. 

\subsection{An Influence Tester for Unitaries}
\label{subsec:testable-influence}

% \red{We start by giving an alternate (and for our purposes, more useful) characterization of the influence of a qubit on a unitary. This characterization admits an efficient ``tester,'' \QIT, which is described below in Algorithm~\ref{alg:qit}. Looking ahead, this subroutine will be crucial to the quantum junta tester that we describe in \Cref{subsec:ambainis}.

% With this alternative characterization of influence, we can compute an unbiased estimator for $\Inf_S[U]$ using \Cref{alg:raw-IT}.}

We start by describing a subroutine \RIT (cf. \Cref{alg:raw-IT}) that allows us to estimate the influence of a set of variables $S\sse[n]$ on a unitary $U$. 

% \begin{algorithm} 
% \caption{Influence Estimator for Quantum Unitaries} 
% \label{alg:raw-IT}
% \vspace{0.5em}

% \textbf{Input:} Oracle access to $U$ and $ U^\dagger$ for $U\in \calU_N$, $S \subset [n]$ \\[0.5em]
% \textbf{Output:} $X \in \{ 0, 1\}$ 

% \

% \RIT$\pbra{U,S}$:

% \begin{enumerate}
% 	\item Uniformly draw $i \sim \zo^{n-|S|}$, $b_1, b_2 \sim \zo^{|S|}$, and prepare the state $\ket{\psi} := \ket{i}_{\overline{S}}\ket{b_1}_S$.
% 	\item Apply the unitary $U$ to obtain $U\ket{\psi}$. Measure the qubits in the $S$ register in the computational basis to obtain $b' \in \zo^{|S|}$.
% 	\item If $b' \neq b_1$, output $1$; otherwise
% 	\begin{enumerate}
% 		\item Discard the qubits in register S of $\ket{\psi}$ and prepare $\ket{\phi} := \ket{\psi}_{\overline{S}}\ket{b_2}_S.$ 
% 		\item Apply the unitary $U^\dagger$ to obtain $U^\dagger\ket{\phi}$. 
% 		\item Measure all the qubits in the computational basis to obtain $b'' \in \zo^n$. 
% 		\item If $b'' = (i)_{\overline{S}}\circ (b_2)_{S}$, then output $0$. Otherwise, output $1$.
% 	\end{enumerate}
% \end{enumerate}

% \end{algorithm}

\begin{algorithm} 
\caption{Influence Estimator for Quantum Unitaries} 
\label{alg:raw-IT}
\vspace{0.5em}

\textbf{Input:} Oracle access to $U\in \calU_N$, $S \subset [n]$ \\[0.5em]
\textbf{Output:} $X \in \{ 0, 1\}$ 

\

\RIT$\pbra{U,S}$:

\begin{enumerate}
	\item Prepare the Choi-Jamiolkowski state $\ket{v(U)}$ given by 
	\[\ket{v(U)} = \frac{1}{\sqrt{N}}\sum_{0\leq i,j < N} U[i,j]\ket{i}\ket{j}.\]
	This is prepared by querying $U$ once on the maximally entangled state. 
	\item Measure the $2|S|$ qubits in the registers $S \cup \{\wt{\ell} : \ell\in S\}$ in the Bell basis, $\{ \ket{v(\sigma_x)}\}_{x \in \Z_4 ^{n}}$, \\  and let $\ket{\phi}$ denote the post-measurement state. 
    \begin{enumerate}
        \item Test if $\ket{\phi}$ is equal to $\ket{\mathrm{EPR}}^{\otimes |S|}$, return 0.
        \item Otherwise, return $1$.
    \end{enumerate}
\end{enumerate}

\end{algorithm}

\begin{lemma}\label{lemma:raw-tester-expectation}
Let $X$ denote the output of \RIT$(U, S)$\ for $U\in\calU_N$ and $S\sse[n]$ as described in \Cref{alg:raw-IT}. Then 
\[\E[X] = \Inf_S[U].\]
\end{lemma}

\begin{proof}
Recall that $U$ can be written in the Pauli basis as $U = \sum_{x\in \Z_4^n} \wh{U}(x)\sigma_x$. Thus, $\ket{v(U)}$ can be written as 
\begin{align*}
    \ket{v(U)} &= \sum_{x\in \Z_4^n} \wh{U}(x)\ket{v(\sigma_x)} \\
    &= \sum_{x : \supp(x)\cap S = \emptyset} \wh{U}(x)\ket{v(\sigma_x)} + \sum_{x : \supp(x)\cap S \neq \emptyset} \wh{U}(x)\ket{v(\sigma_x)}\\
    &= \sum_{x : \supp(x)\cap S = \emptyset} \wh{U}(x)\ket{v(\sigma_{x_{\overline{S}} })} \ket{v(I^{\otimes |S|})} + \sum_{x : \supp(x)\cap S \neq \emptyset} \wh{U}(x)\ket{v(\sigma_{x_{\overline{S}} })}\ket{v(\sigma_{x_{S} })}.
\end{align*}
Where $x_S \in \Z_4^S$ is notation for the restriction of $x$ onto the qubits in $S$. Similarly, $\sigma_{x_S}$ is the Pauli basis vector given by the tensor product of $| S |$ Pauli matrices according to $x_S$. Thus, for any $x\in\Z_4^n$ such that $\supp(x)\cap S \neq \emptyset$, the state $\ket{v(\sigma_{x_S})}$ is orthogonal to  the state $\ket{v(I^{\otimes |S|})} = \ket{\mathrm{EPR}}^{\otimes |S|}$. Because $\{ |\wh{U}(x)|^2\}_{x\in \Z_4^n}$ forms a probability distribution, when \Cref{alg:raw-IT} measures the qubits in $S \cup \{\wt{\ell} : \ell \in S\}$, it will return 1 with the following probability.
\begin{align*}
    \Ex[X] &= \Pr[X = 1]\\
    &= \sum_{x : \supp(x)\cap S \neq\emptyset} |\wh{U}(x)|^2\\
    &= \Inf_S[U]
\end{align*}
This completes the proof. 
\end{proof}

Note that we can boost the probability that \RIT\ outputs $1$ via amplitude amplification (see, for example, Section~2.2 of \cite{Montanaro2010}).
%\gray{
%To implement amplitude amplification, we would first need to purify the circuit by adding a small number of ancilla qubits, so that the circuit acts on a larger number of qubits but does not perform any measurements until the end. This rewiring is guaranteed by the Deferred Measurement Principle\cite{}[\thomas{should ask if this is the right way to put it} \shivam{Let's run it by Henry! I'm copying Montanaro and saying that this is ``standard.''}]. However, this rewiring would not affect the number of queries that each call to Raw-Unitary-Influence-Tester makes to the oracle $U$ or $U^\dagger$.
%}
In particular, we can amplify the probability of \RIT\ outputting $1$ from $\delta$ to an arbitrary constant (say 0.9) via $O(1/\sqrt{\delta})$ calls to the oracles for the unitary $U$. Thus, we have the following lemma.

\begin{algorithm}
\caption{Influence Estimator via Amplitude Amplification}
\label{alg:qit}
\vspace{0.5em}
\textbf{Input:} Oracle access to $U\in\calU_N$, $S \sse [n]$, $\delta \in (0, 1]$ \\[0.5em]
\textbf{Output:} $X \in \{ 0, 1\}$ 

\

\QIT$\pbra{U, S, \delta}$:
\begin{enumerate}
	\item Use amplitude amplification with $O(1/\sqrt{\delta})$ calls to \RIT $(U, S)$. 
	\item Return the same value as \RIT $(U, S)$.
\end{enumerate}
\end{algorithm}

\begin{lemma}
Let $U\in\calU_N$ and $S\sse[n]$. If $\Inf_S[U]\geq\delta$, then $\QIT(U,S,\delta)$ as described in \Cref{alg:qit} outputs $1$ with probability at least 9/10, and if $\Inf_S[U] = 0$, then $\QIT(U,S,\delta)$ always outputs $0$. Furthermore, the number of queries made to $U$ is $O(1/\sqrt{\delta})$.
\end{lemma}

\subsection{Reducing to Gapped Group Testing}
\label{subsec:ambainis}

Using our influence estimator \QIT, we can now reduce the problem of testing quantum juntas to that of Gapped Group Testing (GGT), which we define below. Our approach closely follows that of Ambainis et al. \cite{Ambainis2016}, who reduce the problem of testing $k$-juntas to GGT. We remark that certain  parameters in our adaptation of Ambainis et al.'s algorithm will be worse by a square-root factor, resulting in an overall query complexity of $\wt{O}({\sqrt{k}}/{\eps})$ for testing quantum $k$-juntas as opposed to $\wt{O}({\sqrt{k/\eps}})$ as obtained by Ambainis et al. for testing classical juntas.

We first define the exact version of Group Testing.

\begin{definition}[EGGT] \label{def:eggt}
	Let $k$ and $d$ be positive integers, $\calX$ consist of all subsets of $[n]$ with size $k$, and $\calY$ consist of all subsets of $[n]$ of size $k + d$. In the Exact Gapped Group Testing (EGGT) problem, we are given oracle access to the function $\mathrm{Intersects}_A, A \in \calX \cup \calY$ and must decide whether $A \in \calX$ or if $A \in \calY$
\end{definition}

The exact GGT will be referenced in the analysis. However, the actual algorithm we will use in our algorithm solves a more general version of EGGT.

\begin{definition}[GGT] \label{def:ggt}
	Let $k$ and $d$ be positive integers. Define two families of functions 
	\[\wt{\calX} = \cbra{f\isazofunc \left. \right\vert \: \exists A \in \calX \: \forall S \subset [n] : S \cap A = \emptyset \implies f(S) = 0}\]
	\[\wt{\calY} = \cbra{f\isazofunc \left. \right\vert \: \exists B \in \calY \: \forall S \subset [n] : S \cap B \neq \emptyset \implies f(S) = 1}\]
	In an instance of \GGT$(k,d)$, given oracle access to some function $f \in \wt{\calX}\cup\wt{\calY}$, decide whether $f\in \wt{\calX}$ or $f\in\wt{\calY}$.
\end{definition}

Note that if the function $f$ is in $\wt{\calX}$, then sets $S$ such that $S \cap A \neq \emptyset$ do not restrict $f$. They are ``irrelevant." Similarly, if $f$ is in $\wt{\calY}$, sets $S$ such that $S \cap B = \emptyset$ are ``irrelevant'' (cf. \Cref{sec:irrel-var}). More precisely, the sets that are deemed irrelevant follow from the adversary bound and is explained in more detail in Observation 3.9 of \cite{Ambainis2016}.
%\tnote{Want to make the connection here to irrelevant variables clearer}
Also, note that if we replace implication symbols in \Cref{def:ggt} with equivalence symbols, we recover the EGGT problem. Thus, EGGT is a special case of GGT.

To get some intuition for \Cref{def:ggt}, consider the following scenario: Given $n$ soldiers, some of which are sick, you would like to determine whether there are at most $k$ sick soldiers, or if there are at least $k+d$ sick soldiers. You are allowed to test this by pooling blood samples from subsets of the $n$ soldiers, where the pooled test returns positive if the group contains at least one sick soldier.

More precisely, for an unknown $A\sse [n]$, we would like to decide if $|A|\leq k$ or $|A|\geq k+d$ given access to the following oracle
\[\mathrm{Intersects}_A(S) := \begin{cases}
 1 & A\cap S \neq \emptyset\\ 0 & \text{otherwise}	
 \end{cases}.
\]

We briefly explain the connection to junta testing: Given a unitary $U$ and a fixed threshold $\delta >0$, let $S_\delta \sse[n]$ be the set of qubits whose influence is at least $\delta$. Note then that $\QIT(U,T,\delta)$ will will return 1 with high probability if at least one variable in $S$ is in $S_{\delta}$. In this sense, we have that 
\[\QIT(U,T,\delta) \qquad \approx \qquad \mathrm{Intersects}_{S_{\delta}}(T)\]
By examining various settings of $\delta$, we can use GGT to infer the ``distribution" of influence of a unitary U among its qubits. We will make use of the following quantum algorithm obtained by Ambainis et al. for GGT. 

\begin{theorem}[Theorem 3.6 of \cite{Ambainis2016}] \label{thm:ggt-ambainis}
	There exists a quantum algorithm \QGGT\ that solves \GGT$(k,d)$ using $O(\sqrt{1 + k/d})$ queries. 
\end{theorem}

\begin{algorithm} 
\caption{Quantum $k$-Junta Tester}
\label{alg:qjt}
\vspace{0.5em}
\textbf{Input:} Oracle access to $U$, parameter $k$ \\[0.5em]
\textbf{Output:} ``Yes'' or ``No'' 

\

\QJT$\pbra{U,k}$
\begin{enumerate}
	\item Run \TOne$(U, k, l)$ for $l \in \{0,\ldots, \lfloor \log(200k) \rfloor\}$. 
	\item Run \TTwo$(U, k)$.
	\item Output ``Yes'' if all $ \lfloor \log(200k) \rfloor + 2$ testers above accept, and output ``No'' otherwise. 
\end{enumerate}

\end{algorithm}

\begin{algorithm}[t] 
\caption{Tester of the First Kind}
\label{alg:first-tester}
\vspace{0.5em}
\textbf{Input:} Oracle access to $U$, parameter $k$, parameter $l$\\[0.5em]
\textbf{Output:} ``Yes'' or ``No'' 

\ 

\TOne$(U, k, l)$:
\begin{enumerate}
    \item Let $d_l = 2^l$ and $\delta_l = \frac{\epsilon^2}{2^{l + 5} \log(400k)}.$
    \item Run \QGGT\ with parameters $k$ and $d = d_l$, and query access to the following oracle:
    \[\text{Given }S\sse[n]\text{, output } \QIT(U,S,\delta_l)\]
    \item Output ``Yes'' if GGT accepts, and ``No'' otherwise.
\end{enumerate}

% Run Gapped Group Testing (GGT) algorithm with parameters k and $d_l = 2^l$ and the following Unitary-Influence-Tester oracle \;
% \begin{itemize}
%     \item On input $S \subset [n]$, run Unitary Influence tester on $V = S$, $\delta_l = \frac{\epsilon^2}{2^{l + 5} \log(400k)}$
% \end{itemize}
% Accept if GGT accepts, otherwise reject

\end{algorithm}

\begin{algorithm}
\caption{Tester of the Second Kind}
\label{alg:second-tester}
\vspace{0.5em}
\textbf{Input:} Oracle access to $U$, parameter $k$\\[0.5em]
\textbf{Output:} ``Yes'' or ``No'' 

\ 

\TTwo$(U, k)$:
\begin{enumerate}
    \item Estimate acceptance probability of following subroutine up to additive error 0.05:
    \begin{itemize}
        \item Generate $\bS \subset [n]$ by adding $i\in [n]$ to $\bS$ with probability $1/k$ independently.
        \item Run $\QIT(U, \bS, \delta)$ where $\delta := \frac{\eps^2}{16k}$.
    \end{itemize}
    \item Output ``Yes'' if estimated acceptance probability is at most 0.8, and ``No'' otherwise.
\end{enumerate}

% Estimate acceptance probability of following subroutine with error 0.05:
    
%     \begin{itemize}
%         \item Generate $V \subset [n]$ by adding i to V with probability $1/k$ independently at random
%         \item Run Unitary Influence Tester with V and $\delta = $
%     \end{itemize}
    
% Accept if the estimated acceptance probability $ \leq 0.8$. Otherwise reject.

\end{algorithm}

Our algorithm for quantum junta testing and analysis thereof closely follow the structure of Ambainis et al.'s algorithm for junta testing and its analysis; we include complete details below for completeness but refer the interested reader to Section~4 of \cite{Ambainis2016} for the original algorithm. 

Note that because our \QIT\ serves as a subroutine to the GGT algorithm, there is a need for a careful analysis of the properties of their composition. This is addressed in~\Cref{sec:irrel-var} at a high level and addressed in more detail in~\cite{Ambainis2016}.

\begin{theorem} \label{thm:main-tester}
Given $U \in\calU_N$, with high probability $9/10$,  the algorithm \QJT$(U)$  outputs ``Yes'' if $U$ is a $k$-junta, and outputs ``No'' if $U$ is $\eps$-far from every quantum $k$-junta. Furthermore, \QJT$(U)$ makes $O\pbra{\frac{\sqrt{k\log k}}{\eps}\log k}$ calls to the unitary $U$ and has two-sided error. 
% Quantum-$k$-Junta-Tester distinguishes quantum k-juntas from unitaries that are $\eps-$far from any quantum k-junta with query complexity 
\end{theorem}

\begin{proof}

The setup and analysis of the algorithm (\Cref{lem:far-then-one-of-two-cases,lem:TOne-behavior,lem:alg-second-kind}) is almost the same as in \cite{Ambainis2016}, with a few constants changed. 

Without loss of generality, we assume that the first $K$ qubits are the most influential ones and are ordered in decreasing amount of influence.
\[\Inf_1[U] \geq \Inf_2[U] \geq \ldots \geq \Inf_K[U] > 0 = \Inf_{K+1}[U] = \ldots = \Inf_n[U].\]
    Of course, the tester does not know this order. The primary challenge is in showing that if $U$ is $\eps$-far from every quantum $k$-junta, then at least one of the two subroutines \TOne\ and \TTwo\ will output ``No'' with significant probability. The $\lfloor \log(200k) \rfloor + 2$ tests in the main \Cref{alg:qjt} are tailored for this purpose; in particular, we have the two following cases when $U$ is $\eps$-far from every quantum $k$-junta:
    \begin{enumerate}
        \item \textbf{Case 1:} $\sum_{j= k + 1}^{200k} \Inf_j[U] \geq \eps^2/8$. This case is further split into $\lfloor \log 200k \rfloor + 1$ subcases:
    \begin{align*}
        \left|\left\{ j \in [n] : \Inf_j[U] \geq \frac{\epsilon^2}{2^{l + 5} \log(400k)}\right\}\right| \geq k + 2^l
    \end{align*}
    for $l \in \{ 0, ..., \lfloor \log(200k) \rfloor \}$. We say that a unitary $U$ is a \emph{non-junta of the first kind}  if this is the case for some $l \in \{ 0, ..., \lfloor \log(200k) \rfloor\}$.
        \item \textbf{Case 2:} $\sum_{j= k + 1}^{200k} \Inf_j[U] \leq  \eps^2/8$. We say $U$ is a \emph{non-junta of the second kind} if this is the case. 
    \end{enumerate}
    
    \Cref{lem:far-then-one-of-two-cases} says that any unitary $U$ that is $\eps$-far from every quantum $k$-junta satisfies at least one of the two cases above. The correctness and query complexity of \Cref{alg:qjt} now follows from \Cref{lem:far-then-one-of-two-cases,lem:TOne-behavior,lem:alg-second-kind}.
\end{proof}

Finally, we prove the auxiliary lemmas used in the proof of the above theorem. \Cref{lem:far-then-one-of-two-cases,lem:TOne-behavior,lem:alg-second-kind} are analogous to Lemmas 4.3 to 4.5 of \cite{Ambainis2016}.

\begin{lemma} \label{lem:far-then-one-of-two-cases}
Every $U$ that is $\eps$-far from being a quantum $k$-junta satisfies one of the two cases above.
\end{lemma}

\begin{proof}
 It suffices to show that if $U$ is a non-junta of the first kind, then at least one of the $\lfloor \log 200k \rfloor + 1$ sub-cases holds. By definition, we have 
 \[\sum_{j= k + 1}^{200k} \Inf_j[U] \geq \eps^2/8.\]
 Define 
 \[\eps' = \frac{\eps^2}{32 \log (400k)}\]
 and consider the partition of $[0,1]$ given by 
 \[A_{\infty} = \left[0, \frac{\eps'}{2^{\lfloor \log 200k \rfloor}} \right), \quad A_0 = [\eps', 1], \quad A_{l} = \left[\frac{\eps'}{2^l}, \frac{\eps'}{2^{l - 1}}\right) \]
 where $l\in \{\lfloor \log 200k \rfloor, \ldots, 1\}$. Define $B_l := \{ j \in \{k + 1,..., 200k\} : \Inf_j[U] \in A_l\},$
 and note that each $j\in[n]$ is included in exactly one of the $B_l$. Writing 
 \[W_l = \sum_{j \in B_l} \Inf_j[U] \qquad\text{we have}\qquad \sum_l W_l \geq \eps^2/8\] as $U$ is a non-junta of the first kind. We also have that \[W_\infty < 200k \pbra{\frac{\eps^2}{32 \cdot 2^{\lfloor \log 200k \rfloor}}} < \frac{\eps^2}{16}.\] Since the maximum of the $W_l$'s is at least their average, there exists $l^\ast \in \{ 0, 1, ... \lfloor \log 200k \rfloor \}$ such that \[W_{l^\ast} \geq \frac{\eps^2}{16 \log 400k},\] which in turn implies
\begin{align*}
|B_l| \geq \frac{\frac{\eps^2}{16 \log 400k}}{\frac{\eps^2\cdot 2^{1-l}}{32 \log 400k}} = 2^l.
\end{align*}

Every variable $j \in B_l$ has influence at least $\frac{\eps'}{2^l} =: \delta_l$. Furthermore, since the influence of variables are ordered in decreasing order, each variable $j \in [k]$ also has at least $\delta_l$ influence. Thus, there are at least $k + 2^l$ indices $j$ such that $\Inf_j[U] \geq \delta_l$, and $U$ satisfies the first case for this particular $l$.

\end{proof}

\begin{lemma} \label{lem:TOne-behavior} If $U$ is a $k$-junta, then all calls to \TOne\ will accept with high probability. If $U$ is a non-junta of the first kind, then one of the calls to \TOne\ will reject with high probability. Finally, the overall query complexity of all $\lfloor \log 200k \rfloor + 1$ testers of the first kind is \[O\pbra{\frac{\sqrt{k \log k}}{\eps} \log k}.\]
\end{lemma}

\begin{proof}

The composition in \TOne\ is done as described in Definition 2.10 of \cite{Ambainis2016} which allows for a tight query-complexity. Towards this definition, $F$ and $(G_j)$ are defined as follows: The partial function $F$ is the EGGT function from \Cref{def:eggt}. $F$ takes in a function $h$ and outputs $0$ if $h = \mathrm{Intersects}_A, |A| = k$ and $1$ if $h = \mathrm{Intersects}_A, |A| = k + d$. In other cases, $F$ is undefined. For each $S \subset [n]$, the partial function $G_S$ is our \QIT\ on set S. $G_S$ is partial in that it equals 1 if $\Inf_S[U] \geq \delta$, equals 0 if $\Inf_S[U] = 0$, but is undefined for anything in between. Thus, \TOne\ is equivalent to the following composition:
\begin{align}
    U \rightarrow (G_{\emptyset}(U), G_{\{ 1\}}(U), G_{\{ 2\}}(U), ... G_{[n]}(U))
\end{align}
The irrelevant variables to the function $F$ correspond to the sets $S$ that do not impact its output; that is, whatever \QIT\ outputs on these sets do not matter to $F$. Because we use the same GGT algorithm, derived from the same solution to the adversary bound as \cite{Ambainis2016}, we have the same irrelevant variables. 

\begin{enumerate}
    \item If the input A is in $\calX (|A| = k)$, a set $S \subset [n]$ is irrelevant if $S \cap A \neq \emptyset$. That is, if $U$ is a $k$-Junta, \TOne\ only looks at sets such that $S \cap A = \emptyset$.
    \item If the input A is in $\calY (|A| = k + d)$, a set $S \subset [n]$ is irrelevant if $|S \cap A| \neq 1$. In particular, if $U$ is $\eps-$far from a $k$-Junta, \TOne\ ignores sets such that $S \cap A = \emptyset$
\end{enumerate}

Suppose $U$ is a non-junta of the first kind, satisfying case $l$, in the sense of Lemma \ref{lem:far-then-one-of-two-cases}. By definition, there is an $A \subset [n], |A| = k + 2^l$ such that for all $j \in A$, $\Inf_j[U] \geq \delta_l$. By the monotonicity of influence, $\Inf_S[U] \geq \delta$ for all $S$ that intersect A. Finally, because the sets that are disjoint from A are irrelevant in the non-junta case, $\TOne$'s oracle behaves like an $\mathrm{Intersect_A}$ oracle that depends on at least $k + 2^l$ indices. Thus, this instantiation of \Cref{alg:first-tester}'s GGT will reject with high probability. 

Finally, if $U$ is a $k$-junta, then there is a set $A \subset [n], |A| \leq k$ such that if $S \cap A = \emptyset$, then $\Inf_S[U] = 0$. Because all sets $S \cap A \neq \emptyset$ are irrelevant in the $k$-junta case, \TOne's oracle behaves like an $\mathrm{Intersect_A}$ oracle that depends on $k$ indices. Thus, all the $\TOne$'s will accept with high probability as there are at most $k$ influential variables. 

Thus, the tester of the first kind, a group tester instantiated with $d = 2^l$ and $\delta_l$, will be able to distinguish between this case from case where $U$ is a $k$-junta, where the set of variables of influence at least $\delta_l$ is size at most $k$.

Finally, for a particular value of $l$, the query complexity of the influence tester is $O({\delta_l}^{-1/2})$ while the query complexity of the corresponding group tester instance is $O(\sqrt{k/d_l})$. It then follows by \Cref{lemma:fn-composition} that the complexity of any tester of the first kind is 
\[O\pbra{\sqrt{\frac{k}{2^l}}\cdot \sqrt{\frac{2^l \log 400k}{\eps^2}}} = O\pbra{\frac{\sqrt{k \log k}}{\eps}}\]
giving an overall query complexity of 
\[O\pbra{\frac{\sqrt{k \log k}}{\eps} \log k}\]
for all $\lfloor \log(200k) \rfloor + 1$ testers of the first kind. 
\end{proof}

\begin{lemma} \label{lem:alg-second-kind}
\Cref{alg:second-tester} accepts if $U$ is a $k$-junta and rejects if $U$ is a non-junta of the second kind, and its query complexity is $O(\sqrt{k/\eps})$
\end{lemma}

\begin{proof}
We show that the procedure described in Item~1 of \Cref{alg:second-tester} has acceptance probability \emph{at most} $ 0.75$ if $U$ is a k-junta, and has acceptance probability \emph{at least} $ 0.85$ if $U$ is a non-junta of the second kind.

Suppose $U$ is a $k$-junta. Then the probability that the set $S$ does not intersect the set $J$ of relevant variables is 
\[\left(1 - \frac{1}{k}\right)^{|J|} \geq \pbra{1 - \frac{1}{k}}^{k} \geq \frac{1}{4}.\]
Therefore, with probability at least $0.25$, we have $S\cap J = \emptyset$ in which case $\Inf_{S}[U] = 0$. It follows then that the acceptance probability above is at most 0.75. 

Now suppose $U$ is a non-junta of the second kind. For $j \in [n]$, define
\begin{align*}
\UInf_j[U] :=  \begin{cases} 
      0 &   j \leq 200k \\
      \sum_{x : \supp(x) \cap \{ 200k + 1 ... j\} = \{ j\} } |\wh{U}(x)|^2 & \text{otherwise}
   \end{cases}.
\end{align*}
For $S \subset [n]$, define $\UInf_S[U] := \sum_{j \in S} \UInf_j[U]$. It is easy to see that 
\[\UInf_S[U] \leq \Inf_S [U],\] 
and that for $S, T\sse [n]$ with $S\cap T = \emptyset$, we have 
\[\UInf_{S \cup T}[U] = \UInf_S[U] + \UInf_T[U].\]
Now, because $U$ is $\eps$-far from every quantum $k$-junta, by Lemma \ref{lemma:eps-far-guarantees}, we have that
\begin{equation} \label{eq:banana-1}
    \Inf_{\{k + 1... K\}} [U]\geq \eps^2/4,
\end{equation}
and since $U$ is a non-junta of the second kind,
\begin{equation} \label{eq:banana-2}
    \sum_{j = k + 1}^{200k} \Inf_j[U] \leq \frac{\eps^2}{8}.
\end{equation}
Combining \Cref{eq:banana-1,eq:banana-2} we get that
\begin{align}
    \UInf_{[n]}[U] &= \Inf_{\{200k + 1 ... K\}} [U] \nonumber \\
&\geq \Inf_{\{k + 1... K\}} [U] - \sum_{j = k + 1}^{200k} \Inf_j[U] \nonumber \\
&\geq \frac{\eps^2}{8}. \label{eq:banana-3}
\end{align}

Consider now the random variable $\UInf_{\bS}[U]$ where $\bS$ is drawn as described in \Cref{alg:second-tester}. We have 
\[\mu:= \mathop{\mathbb{E}}_{\bS}\sbra{\UInf_{\bS}[U]} = \frac{1}{k}\cdot \UInf_{[n]}[U] \geq \frac{\eps^2}{8k}.\]
We also have 
\begin{align*}
    \sigma^2 := \Var\sbra{\UInf_{\bS}[U]} &\leq \frac{1}{k} \sum_j \UInf_j[U]^2\\
&\leq \frac{1}{k} \left( \max_j \UInf_j[U]\right)\cdot  \UInf_{[n]}[U]\\
&\leq \frac{1}{k} \cdot\left(\frac{\eps^2}{4 \cdot 200k}\right) \UInf_{[n]}[U]\\
&\leq \frac{\mu^2}{100}.
\end{align*}
% \begin{align}
% \sigma^2 &= Var[\UInf_V[U]]\\
% &\leq \frac{1}{k} \sum_j \UInf_j[U]^2\\
% &\leq \frac{1}{k} ( \max_j \UInf_j[U]) \UInf_{[n]}[U]\\
% &\leq \frac{1}{k} \frac{\eps^2}{4 \cdot 200k} \UInf_{[n]}[U]\\
% &\leq \mu^2/100
% \end{align}
It then follows by Chebyshev's inequality that 
\begin{align}
    \Pr\left[\UInf_{\bS}[U] < \frac{\eps^2}{16k}\right] &\leq \Pr\sbra{|\UInf_{\bS}[U] - \mu| > \frac{\mu}{2}}\\
    &\leq \Pr\sbra{|\UInf_{\bS}[U] - \mu| > 5\sigma}\\
    &\leq \frac{1}{25}.
\end{align}
In other words, the probability $\UInf_{\bS}[U] > \eps^2/16k$ is at least 0.96. So the acceptance probability of the subroutine described in Item~1 of \Cref{alg:second-tester} on $\bS$ is at least $0.9 \times 0.96 > 0.85$ if $U$ is a non-junta of the second kind. 

Finally, the subroutine of the tester of the second kind only makes queries to $\QIT$\ on $\delta = \frac{\eps^2}{16k}$, which requires complexity $O(\sqrt{k/\eps^2})$, and the outer estimation overhead is a constant. 
\end{proof}

\section{An $\Omega(\sqrt{k})$ Lower Bound for Testing Quantum $k$-Juntas}
\label{sec:lower-bound}

In this section, we obtain an $\Omega(\sqrt{k})$ lower bound for testing quantum $k$-juntas, which shows that the algorithm obtained in \Cref{sec:testing-ub} is essentially optimal (up to polylogarithmic factors in $k$). Our lower bound follows via a natural reduction from testing classical $k$-juntas to testing quantum $k$-juntas, combined with the $\Omega(\sqrt{k})$ lower bound for testing classical $k$-juntas obtained by Bun, Kothari, and Thaler \cite{Bun2017}.
The key technical insight here is in \Cref{lemma:lb-structural-result}, which shows that every quantum $k$-junta is (in a certain sense) ``close'' to a quantum Boolean function (i.e. a Hermitian quantum unitary).

%\gray{The lower bound for quantum $k$-junta testing follows from a reduction from the $\Omega(\sqrt{k})$-query lower bound obtained by Bun et al. for testing classical $k$-juntas via quantum algorithms \cite{Bun2017}. This is a fairly natural reduction, except that the notion of distance is  across the two scenarios.}

In what follows, we say that an algorithm is a \emph{$(k, \eps)$-classical} (respectively \emph{quantum}) \emph{junta tester} if, given query access to a Boolean function $f\isazofunc$ (respectively unitary $U\in\calU_N$), with probability at least $9/10$ it outputs 
\begin{itemize}
	\item ``Yes'' if $f$ (respectively $U$) is a $k$-junta; and
	\item ``No'' if $f$ (respectively $U$) is $\eps$-far from every $k$-junta. 
\end{itemize}

\begin{theorem} \label{thm:reduction}
Every $T$-query $(k,\sqrt{\eps/2})$-quantum junta tester is also a $T$-query $(k, \eps)$-classical junta tester.  
\end{theorem}

Note that \Cref{thm:reduction} together with the $\Omega(\sqrt{k})$ lower bound for quantum testing of $k$-juntas by Bun, Kothari, and Thaler~\cite{Bun2017} implies the desired lower bound. Before proving \Cref{thm:reduction}, we first introduce some notation. Given a Boolean function $f\isazofunc$, we will write 
\begin{equation} \label{eq:tomato}
	U_f := \mathrm{diag}\pbra{(-1)^{f(x)}}
\end{equation}
as a diagonal matrix whose diagonal entries are the $2^n$ values of the function $f$. Note that $U_f$ is unitary as its singular values are $\pm 1$. Also, given a matrix $A$, we will use $A[i,j]$ to mean the entry of $A$ at row $i$ and column $j$.

The transformation we use to reduce Boolean functions to quantum Boolean functions (towards the goal of proving \Cref{thm:reduction}) is the natural one given by \Cref{eq:tomato}. First, if a function $f\isazofunc$ is a $k$-junta, then $U_f$ is also a quantum $k$-junta. To see this, suppose without loss of generality that the last $k$ bits of $f$ are the relevant ones\footnote{We will use this indexing convention for the remaining sections as well.}, i.e. we have 
\[f(x) = \widetilde{f}(x_{n-k+1}, \ldots, x_n)\] for some $\widetilde{f}: \zo^k \to \zo$. It then follows that
\[U_f = I^{\otimes (n - k)} \otimes U_{\wt{f}}.\]

The following lemma shows that an analogous statement holds when $f$ is far from being a $k$-junta, from which \Cref{thm:reduction} is immediate.

\begin{proposition} \label{prop:lb-no-case}
	If $f\isazofunc$ is $\eps$-far from every $k$-junta, then $U_f$ is $\sqrt{\epsilon/2}$-far from every quantum $k$-junta.
\end{proposition}

\begin{proof}
	We will first show if $g\isazofunc$ is a $k$-junta, then 
	\begin{equation} \label{eq:lb-1}
		\dist(U_f, U_g) \geq \sqrt{2\epsilon}.	
	\end{equation}
	As $f$ is $\eps$-far from $g$, we have that 
	\[\Pr\sbra{f \neq g} \geq \eps.\]
	Consider $U_g$, the unitary whose diagonal entries are the values of $g$, as we did with $f$ above. The distance between $U_g$ and $U_f$ is at least
\begin{align}
\dist(U_f, U_g) ^2 &= \frac{1}{2N} \cdot \min_\theta   \| e^{i \theta} U_f - U_g\|^2 \nonumber\\
&= \min \left(\frac{1}{2N} \| U_f - U_g \|^2, \frac{1}{2N} \| -U_f - U_g \|^2\right) \label{eq:realsmatter}\\
&= \min\pbra{\frac{2}{N} \sum_{x\in\zo^n} \pbra{\frac{f(x) - g(x)}{2}}^2, \frac{2}{N} \sum_{x\in\zo^n} \pbra{\frac{f(x) + g(x)}{2}}^2} \nonumber\\
&= 2 \min \pbra{\Pr[f \neq g], \Pr[f = g]} \nonumber\\
&\geq 2\eps.\label{eq:kjuntaclosure}
\end{align}

Equation \ref{eq:realsmatter} holds because $U_f$ and $U_g$ are both diagonal with real entries, so the only possible phases that would minimize the Frobenius norm of their difference are $\theta = 0 $ or $\pi$. Equation \ref{eq:kjuntaclosure} holds because $k$-juntas are closed under negation; in more detail, if $g$ is a $k$-junta, then $1-g$ is also a $k$-junta and so \[\Pr[f=g] = \Pr[f \neq 1-g] \geq \eps.\]

We thus have that $\dist(U_f, U_g) \geq \sqrt{2\eps}$. In order to prove the lemma, it suffices to show that for any quantum $k$-junta $V$, there exists a Boolean $k$-junta $g\isazofunc$ such that $\dist(V, U_g) \leq \dist(V, U_f)$. This is proved in \Cref{lemma:lb-structural-result}.

To see why this suffices, note that if this were the case, then by the triangle inequality, 
\[\dist(V, U_f) + \dist(V, U_g) \geq \dist(U_f, U_g).\]
However, as $\dist(U_f, U_g) \geq \sqrt{2\eps}$ by \Cref{eq:kjuntaclosure}, and as $\dist(V, U_g) \leq \dist(V, U_f)$ by \Cref{lemma:lb-structural-result}, we have that 
\[2\cdot \dist(V, U_f) \geq \sqrt{2\eps}\]
and so the result follows.
\end{proof}

%Thus, $D(U_f, U_g) \geq \sqrt{2\eps}$. In order to show that $U_f$ is $\sqrt{\eps/2}$ far from every quantum $k$-junta $V$, it suffices to show that for every quantum $k$-junta $V$ there exists a classical $k$-junta $g$ such that $D(V, U_g) \leq D(V, U_f)$. If this is the case, by triangle inequality, $2 D(V, U_f) \geq D(V, U_f) + D(V, U_g) \geq D(U_f, U_g) \geq \sqrt{2\eps}$

\begin{lemma} \label{lemma:lb-structural-result}
Suppose $f\isazofunc$ is $\eps$-far from every $k$-junta. Then, for every quantum $k$-junta $V\in\calU_N$, there exists some Boolean function $g\isazofunc$ that is a $k$-junta for which 
\[\dist(V,U_g)\leq \dist(V, U_f).\]
\end{lemma} 

\begin{proof}

We can assume without loss of generality that $V$ is a quantum $k$-junta on the last $k$ qubits. We define $g\isazofunc$ as follows: Writing $V = I^{\otimes (n - k)}\otimes\wt{V}$, let 
\begin{equation} \label{eq:choice}
	\widetilde{g} = \arg\min_{h \in \zo^{2^k}} \dist\pbra{\wt{V}, \mathrm{diag}((-1)^{h(x)})} = \arg\min_{h \in \zo^{2^k}} \pbra{ \min_\theta \| e^{i \theta} \wt{V} - \mathrm{diag}((-1)^{h(x)})\|}
\end{equation}
and set $g := \wt{g}$
where we interpret $\wt{g}\isazofunc$ as a $k$-junta.
% let $V$ be an $n$-qubit $k$-junta on the last $k$ qubits. Let $g$ be an $n$-bit $k$-junta on the last $k$ bits that takes the following form.
%
%\begin{align}
%V &=  I_{2^{n - k}} \otimes \tilde{V} \\
%g &= 1_{2^{n - k}} \otimes \tilde{g} \\
%
%\end{align}
We claim that $\dist(U_f, V) \geq \dist(U_g, V)$. First, note that because $U_f$ and $U_g$ are both diagonal matrices, the off-diagonal contributions to 
\[\|U_f - V\|^2 = \sum_{0 \leq i,j < N} |U_f[i,j] - V[i,j]|^2\] is the same as that to $\|U_g - V\|^2$. Moreover, if we multiply the off-diagonal terms of $V$ by a phase $e^{i\theta}$, their contribution to the sum will still be the same as $U_f$ and $U_g$ are zero on their off-diagonal entries. It therefore suffices to compare the diagonal terms of these two quantities. With this in mind, we define the following quantity: For $A, B \in \mathbb{C}^{2^n \times 2^n}$, let $\underline{\text{dist}}(A, B)$ be the sum of diagonal contributions to the Frobenius norm of $A - B$, i.e. we have 
\[
\underline{\text{dist}}(A, B)^2 := \frac{1}{2N}\min_\theta\sum_{0 \leq i < N} |e^{i\theta} A[i,i] - B[i,i]|^2.
\]
We then have that 
\begin{align}
\underline{\text{dist}}(U_f, V)^2 &= \frac{1}{2N}\min_\theta \sum_{j = 0}^{2^{n - k} - 1} \sum_{l = 0}^{2^k - 1} \left|e^{i \theta} V[j \cdot 2^k + l, j \cdot 2^k + l] - U_f[j \cdot 2^k + l, j \cdot 2^k + l]\right|^2 \label{eq:block-expansion}\\
& \geq \frac{1}{2N} \sum_{j = 0}^{2^{n - k} - 1} \min_{\theta_j}   \sum_{l = 0}^{2^k - 1} \left|e^{i \theta_j} V[j \cdot 2^k + l, j \cdot 2^k + l] - U_f[j \cdot 2^k + l, j \cdot 2^k + l]\right|^2 \nonumber\\
&\geq \frac{1}{2^{n - k}} \sum_{j = 0}^{2^{n - k} - 1} \underline{\text{dist}}(U_{\wt{g}}, \wt{V})^2 \label{eq:choice-of-g}\\
&= \underline{\text{dist}}(U_{\wt{g}},\wt{V})^2 \nonumber \\
&= \underline{\text{dist}}(U_g,V)^2 .\nonumber
\end{align}

In particular, \Cref{eq:block-expansion} rewrites the sum by considering $2^{n - k}$ blocks of $2^k \times 2^k$ matrices on the diagonal of $e^{i\theta}V - U_f$, and \Cref{eq:choice-of-g} follows from the choice of $g$ as a minimizer in \Cref{eq:choice}. Because the off-diagonal contributions to the expressions for distance are the same for $U_f$ and $U_g$, $\dist(U_f, V) \geq \dist(U_g, V)$, completing the proof.
\end{proof}

%\begin{theorem}
%$k$-Junta-Testing $\leq_p$ Quantum-$k$-Junta-Testing
%\end{theorem}
%
%In $k$-Junta-Testing, we are given a unitary oracle $\calO_f$ for the unknown boolean function $f$ and are allowed to use quantum algorithms to query the oracle in superposition. First, for a particular n-bit boolean function $f :\{ \pm 1\}^n \rightarrow \{ \pm 1\}$, we represent it as a vector in $\{ \pm 1\}^{2^n}$. Let
%
%\begin{align}
%U_f = diag(f_0, f_1, ... f_{2^n - 1})
%\end{align}
%
%
%Note that $U_f$ is a unitary since its singular values are all $\pm 1$. Furthermore, note that $U_f$ is efficiently preparable as it actually equals $\calO_f$, which is defined by $\calO_f \ket{x} = f(x) \ket{x} \forall x \in \{ 0, 1, ... 2^n - 1\}$. Because $U_f$ is unitary and diagonal, it is its own inverse, satisfying the precondition of the Quantum-$k$-Junta-Testing.
%
%\begin{claim}
%If $f$ is k-junta, then $U_f$ is a quantum k-junta
%\end{claim}
%
%Suppose $f$ is a $k$-junta and WLOG assume that the last k bits are the relevant ones. Then $f = 1_{2^{n - k}} \otimes \widetilde{f}$, for some $\widetilde{f} \in \{ \pm 1\}^{2^k}$, where $1_{2^{n - k}}$ is a vector of ones of length $2^{n - k}$. This implies that $U_f = I_{2^{n - k}} \otimes \widetilde{U}$, where $\widetilde{U} = diag(\widetilde{f})$ is a diagonal matrix on k qubits. That is, $U_f$ is a k-junta.
%
%\begin{claim}
%If $f$ is $\eps-$far from any k-junta, then $U_f$ is $\sqrt{\eps/2}-$far from any quantum k-junta
%\end{claim}

\section{Learning Quantum $k$-Juntas}
\label{sec:learning}

We present algorithm to learn quantum $k$-juntas in \Cref{subsec:learn-ub}, and our lower bound for learning quantum $k$-juntas in \Cref{subsec:learn-lb}.

\subsection{Learning Upper Bound}
\label{subsec:learn-ub}

In this section, we present our algorithm for learning quantum $k$-juntas. Our algorithm can be viewed as analogous to the quantum algorithm of At\i c\i~and~Servedio~\cite{Atici2007} for learning classical $k$-juntas; as such, we start be briefly recalling their high-level approach. 

Given query access to a function $f\isazofunc$, the algorithm of \cite{Atici2007} first determines the set of all relevant variables of non-negligible influence via ``Fourier sampling'' from $f$.\footnote{Recall that Fourier sampling from $f\isazofunc$ refers to drawing $S\sse[n]$ (identified with its $0$/$1$ indicator vector with probability $|\wh{f}(S)|^2$.} It then learns the truth table of the function $f$ restricted to the at most $k$ relevant variables by querying $f$ on each of the $2^k$ possible input strings on the relevant variables. Given membership query access to a unitary $U$, our algorithm proceeds analogously by first learning a set $S$ of relevant qubits with non-negligible influence via ``Pauli sampling'', a subroutine analogous to Fourier Sampling.\footnote{This subroutine is also implicit in \cite{Wang2011}.} Then, we learn an approximation to the part of $U$ that acts only on the subset $S$, the qubits with nonnegligible influence. We do this by reducing the problem to learning a \emph{quantum state}, a task known as \emph{quantum state tomography}.

The connection between learning the unknown unitary $U$ and learning quantum states comes via the \emph{Choi-Jamio\l kowski isomorphism} (described in Section \ref{sec:cj}). In our learning algorithm we will use the following procedure to perform pure state tomography on (copies of) the CJ state $\ket{v(U)}$ in order to learn a description of $U$:

%we use \emph{quantum state tomography} to learn an approximation to $U$ that acts only on $S$, 

% \iffalse
% The following is implicit in Wang~\cite{Wang2011} (who used it to test quantum juntas with $\wt{O}(k)$ queries), and can be viewed as analogous to Lemma~IV.4 of At\i c\i~and~Servedio~\cite{Atici2007}. We also refer the interested reader to the exposition of Wang's algorithm in \cite{Montanaro2016}.

% \begin{proposition}[Algorithm~2 in~\cite{Wang2011}] 
% \label{prop:PS}
% 	Given a quantum $k$-junta $U\in\calU_N$, there exists a procedure \Wang
% 	There exists a procedure \Wang\ such that given query access to a unitary $k$-junta $U\in\calU_N$, it outputs \red{TODO}.
% \end{proposition}

% \fi 

%We will require the following celebrated state tomography algorithm of O'Donnell~and~Wright~\cite{OW16}.

\begin{proposition}[Pure state tomography]
\label{prop:tomography}
    There exists a procedure $\Tomography$ \footnote{This procedure was first devised by Derka, Bu\v{z}ek, and Ekert~\cite{derka1998universal} (and whose sample complexity was determined by Bru{\ss} and Machiavello~\cite{bruss1999optimal}).} that, given $O(d/\eps)$ samples of an unknown $d$-dimensional pure state $\ket{\psi}$, outputs with high probability a classical description of a pure state $\ket{\wh{\psi}} \in \C^d$ such that 
	%Given a density matrix $\rho \in \C^{d\times d}$ of rank $r$ and $\eps > 0$, there exists a procedure $\Tomography$ which, given $O(dr^2/\eps^2)$ samples of $\rho$, outputs with probability at least $99/100$ (a classical description of) a quantum state $\wh{\rho} \in \C^{d\times d}$ such that 
	\[
	    \Big | \left \langle \psi \, \Big | \, \wh{\psi} \right \rangle \Big |^2 \geq 1 - \eps~.
	\]
	%\[\dist(\rho,\wh{\rho}) \leq \eps.\]
\end{proposition}

%We will also use the \emph{Choi-Jamio\l kowski isomorphism} \cite{Choi1975,Jamiolkowski1972} which states that there exists an isomorphism between unitary operators $\calU_N$ and pure states in $\C^{N\times N}$. This allows us to utilize the learning algorithm of \Cref{prop:tomography} to learn quantum unitaries, as described in \Cref{alg:q-learner}. More formally, we have the following.

% \begin{definition} \label{claim:choi} 
% 	Given a unitary $U\in\calU_N$, the \emph{Choi-Jamiolkowski isomorphism} of $U$ is the quantum state given by 
% 	\[\ket{v(U)} := (U \otimes I) \pbra{\frac{1}{\sqrt{N}} \sum_{0 \leq i < N} \ket{i}\ket{i}}  = \frac{1}{\sqrt{N}} \sum_{0 \leq i,j < N} U[i,j] \ket{i} \ket{j}.\]
% \end{definition}

% The Choi-Jamiolkowski isomorphism maps $n-$qubit unitaries to $2n-$qubit states. We can think of the transformation as laying out $n$ EPR pairs and applying the unitary $U$ to half of each EPR pair. 

% \begin{definition} 
% For a $2n$ qubit Choi-Isomorphism state $\ket{\phi}$ and qubit $l \in [n]$, we write $\tilde{l}$ as the qubit originally in an EPR pair with $l$ when $\ket{\phi}$ was synthesized.
% \end{definition}

%We will also require the following easy fact. 

% \begin{fact}
% \label{claim:fidelity}
% For two pure states $\ketbra{\psi}{\psi}$ and $\ketbra{\phi}{\phi}$, the trace norm and fidelity are related as 
% \begin{align*}
%     \| \ketbra{\psi}{\psi} - \ketbra{\phi}{\phi}\|_{\Tr} &= \frac{1}{2}\Tr\pbra{\left|\ketbra{\psi}{\psi} - \ketbra{\phi}{\phi} \right|}\\
%     &= \sqrt{1 - \left|\braket{\psi | \phi} \right|^2}.	
% \end{align*}

% \end{fact}

Our quantum junta learning algorithm is presented in \Cref{alg:q-learner} and its properties are established in \Cref{thm:learning}. 

\begin{algorithm}[t]
\caption{Quantum $k$-Junta Learner}
\label{alg:q-learner}
\vspace{0.5em}
\textbf{Input:} Oracle access to quantum $k$-junta $U$, error parameter $\eps > 0$ \\[0.5em]
\textbf{Output:} Classical description of $U$ (as a $2^n \times 2^n$ matrix)

\

\QLearner$\pbra{U,\epsilon}$:
\begin{enumerate}
    \item Let $S := \PauliSample\pbra{U, \frac{\eps^2}{4k}, k}$.
    % Call \PauliSample$\pbra{U, \frac{\delta^2}{k}, k}$ to learn the set $S \subset [n],|S| \leq k$ of relevant variables with influence at least $\frac{\delta^2}{k}$.
	\item Set $t:= O(\frac{4^k}{\eps^2})$. Call $\QStatePrep(U, S)$ $10t$ times\newline to obtain at least $t$ copies of $\ket{\psi_S}$.
% 	With $t = $, call \QStatePrep($U, \eps/2$, S) $10t$ times to obtain\\ at least $t$ copies of a $2k$-qubit state, $\ket{\psi_U}$ 
	\item Let $\ket{\wh{\psi}} := \Tomography\pbra{\ket{\psi_S}\bra{\psi_S}^{\otimes t},\frac{\eps^2}{4}}$.
% 	Use tomography \cite{OW16} on $\ket{\psi_U}\bra{\psi_U}^{\otimes t}$ to get an estimate $\hat{M}$ that is $\frac{\eps}{4}-$close in trace distance.
	\item Return the unitary encoded by $\ket{\wh{\psi}}$ tensored with $I^{\otimes (n - k)}$
\end{enumerate}
\end{algorithm}

\begin{algorithm}[t]
\caption{The Quantum State Preparation Subroutine (cf. Step~2 of \Cref{alg:q-learner})}
\label{alg:state-prep}
\vspace{0.5em}
\textbf{Input:} Oracle access to unitary $U$, $S\sse [n]$ \\[0.5em]
\textbf{Output:} The quantum state $\ket{\psi_S}$ or ``error''

\

\QStatePrep$\pbra{U, S}$:
\begin{enumerate}
	\item Prepare the state $\ket{v(U)} = \sum_{x \in \Z^n_4} \hat{U} (x) \ket{v(\sigma_x)}$.
	\item  Measure qubits in $\bar{S} \sse [n] $ and $\{ \tilde{l} : l \in \bar{S}\}$ in the Pauli basis $\cbra{ \ket{v(\sigma_x)}}_{x \in \Z_4 ^{n - |S|}}$. 
	\begin{enumerate}
	    \item If the measurement result is $\ket{v(I^{\otimes (n - |S|)})}$, let $\ket{\psi_S}$ be the the unmeasured state \\ on $2|S|$ qubits tensored with $(k - |S|)$ EPR pairs. Return $\ket{\psi_S}$.
	    \item Otherwise, return ``error".
	\end{enumerate}
	
\end{enumerate}
\end{algorithm}

\begin{algorithm}
\caption{The Pauli Sampling Subroutine (cf. Step~1 of \Cref{alg:q-learner})}
\label{alg:pauli-sampling}
\vspace{0.5em}
\textbf{Input:} Oracle access to quantum $k$-junta $U$ on $n$ qubits, threshold $\gamma > 0$ \\[0.5em]
\textbf{Output:} $S\sse[n]$

\

\PauliSample$\pbra{U,\gamma, k}$:
\begin{enumerate}
	\item Initialize $S = \emptyset$.
	\item Repeat the following $O\pbra{\frac{\log k}{\gamma}}$ times:
	\begin{enumerate}
	    \item Prepare the $\ket{v(U)}$ and measure all qubits in the Pauli basis, $\{ \ket{v(\sigma_x)}\}_{x \in \Z_4 ^{n}}$. 
	    \item Given the measurement outcome $\ket{v(\sigma_x)}$, set $S \leftarrow S \cup \supp(x)$
	\end{enumerate}
	\item Return $S$.
\end{enumerate}
\end{algorithm}

\begin{theorem}\label{thm:learning}
Given oracle access to a quantum $k$-junta $U\in\calU_N$ and a $\eps> 0$, \QLearner$(U,\eps)$ (cf. \Cref{alg:q-learner}) outputs, {with probability $9/10$}, a unitary $\wh{U}$ such that $\dist(U, \hat{U}) \leq \eps$. Furthermore, \QLearner\ makes $O\pbra{\frac{k}{\eps} + \frac{4^k}{\eps^2}}$ queries to $U$.
\end{theorem}

\begin{proof}
We will analyze the closeness guarantee and the query complexity separately, starting with the former. 

Consider the state $\ket{\psi_S}$ obtained by running \QStatePrep\ in Step 2 of \Cref{alg:q-learner}. $\ket{\psi_S}$ is a pure state with $2k$ qubits; as such, it encodes $k$-qubit unitary matrix $V$ acting on the qubits in the relevant set $R \subset [n], |R| = k$. We have that 
\begin{equation}
    \label{eq:triangle-1}
    \dist\pbra{U, V \otimes I^{\otimes (n - k)}} \leq \frac{\eps}{2}
\end{equation}
by \Cref{lem:closeness-of-phi-U}. 

Let $\hat{U}$ be the output of Algorithm \ref{alg:q-learner}, and let $\hat{U} := W \otimes I^{\otimes (n - k)}$, for a $k$-qubit unitary $W$ on qubits in the relevant set $R$. To show that $\hat{U}$ is close to $U$, we will now show that with probability at least 99/100, 
\begin{equation}
    \label{eq:triangle-2}
    \dist(V,W) \leq \frac{\eps}{2}.
\end{equation}
It would then follow from the triangle inequality and \Cref{eq:triangle-1,eq:triangle-2} that 
\[\dist(\wh{U},U) \leq \eps.\]

To show that $V$ and $W$ are close, consider the output of \Tomography~in Step~3 of \Cref{alg:q-learner}. By \Cref{prop:tomography}, we have that with $O\pbra{4^k/\eps^2}$ copies of $\ket{\psi_S}$,
\[
\Big | \left \langle \psi_S \, \Big | \, \wh{\psi} \right \rangle \Big |^2 \geq 1 - \eps^2/4~
\]

\iffalse
Let $\ketbra{\phi}{\phi}$ be the nearest rank-$1$ density matrix to $\hat{M}$. We clearly have that 
\[\vabs{ \hat{M} - \ketbra{\phi}{\phi}}_{\Tr} \leq \vabs{ \hat{M} - \ketbra{\psi_U}{\psi_U}}_{\Tr} \leq \frac{\eps}{4}.\]
It then follows by the triangle inequality and \Cref{claim:fidelity} that 
\begin{equation} \label{eq:learn-guarantee}
\vabs{ \ketbra{\psi_U}{\psi_U} - \ketbra{\phi}{\phi}}_{\Tr} \leq \frac{\eps}{2}, \qquad\text{and so}\qquad 
    \left|\braket{\psi_U | \phi}\right| \geq \sqrt{1 - \eps^2/4}.
\end{equation}
\fi 
%Let $\hat{M}$ be the output of the tomography algorithm, step 4 of Algorithm \ref{alg:q-learner}. The guarantee is that with $O(\frac{4^k}{\eps^2})$ copies of $\ket{\psi}$, 
%
%\begin{align}
%     \| \hat{M} - \ketbra{\psi}{\psi}\|_{tr} \leq \eps
%\end{align}
%Let $\ketbra{\phi}{\phi}$ be the nearest rank-1 density matrix to $\hat{M}$. Certainly, we have that $\| \hat{M} - \ketbra{\phi}{\phi}\|_{tr} \leq \| \hat{M} - \ketbra{\psi}{\psi}\|_{tr} \leq \eps$. By triangle inequality and Claim \ref{claim:fidelity}, 
%
%\begin{align}
%    \| \ketbra{\psi}{\psi} - \ketbra{\phi}{\phi}\|_{tr} &\leq 2\eps\\
%    |\bra{\psi}\ket{\phi}| \geq \sqrt{1 - 4\eps^2}
%\end{align}

Note that $\ket{\wh{\psi}}$ encodes $W$ and that $\ket{\psi_S}$ encodes $V$. Writing $K := 2^k$, we have that 
\[\ket{\wh{\psi}} = \sum_{0 \leq i,j < K} \frac{W[i,j]}{\sqrt{K}} \ket{i} \ket{j} \qquad\text{and}\qquad \ket{\psi_S} = \sum_{0 \leq i,j < K} \frac{V[i,j]}{\sqrt{K}} \ket{i} \ket{j}.\]
%\begin{align}
%    \ket{\phi} &= \sum_{0 \leq i,j < K} \frac{W[i,j]}{\sqrt{K}} \ket{i} \ket{j}\\
%    \ket{\psi} &= \sum_{0 \leq i,j < K} \frac{V[i,j]}{\sqrt{K}} \ket{i} \ket{j}
%\end{align}
We then have that 
\begin{align*}
    \dist(V, W)^2 &= \min_\theta \frac{1}{2K} \|e^{i \theta} V - W\|^2\\
    &= \frac{1}{2} \min_{\theta} \sum_{0 \leq i,j < K} \left|\frac{e^{i\theta} V[i,j]}{\sqrt{K}} - \frac{W[i,j]}{\sqrt{K}}\right|^2\\
    &= \frac{1}{2} \min_{\theta} \sum_{0 \leq i,j < K} \pbra{ \left|\frac{e^{i\theta} V[i,j]}{\sqrt{K}}\right|^2 + \left|\frac{W[i,j]}{\sqrt{K}}\right|^2 - 2\cdot\mathfrak{Re}\left(\frac{e^{i\theta}}{K} V[i,j] W[i,j]^*\right) }\\
    &= \frac{1}{2} \min_{\theta} \pbra{ 2 - 2\cdot\mathfrak{Re}\left(\sum_{0 \leq i,j < K} \frac{e^{i\theta}}{K} V[i,j] W[i,j]^*\right) }\\
    &= \frac{1}{2} \pbra{ 2 - 2 \left|\sum_{0 \leq i,j < K}  \frac{1}{K} V[i,j] W[i,j]^\ast\right| }\\
    &= 1 - |\braket{\psi_S\mid \wh{\psi}}|\\
    &\leq \eps^2/4,
\end{align*}

Finally, we turn to the query complexity of \Cref{alg:q-learner}. By \Cref{lem:pauli-sampling}, the query complexity of Step 1 of \Cref{alg:q-learner} (\PauliSample) is 
\[\frac{\log k}{\frac{\eps^2}{4k}} = O\pbra{\frac{k \log k}{\eps^2}}.\]
The number of copies required for the tomography subroutine is 
\[t :=O\pbra{\frac{4^k}{\eps^2}}.\]
As \QStatePrep\ has a small probability of error ($O(\eps^2)$), we can show by Markov's inequality that with $10t$ calls to \QStatePrep, we will obtain at least $t$ copies of $\ket{\psi_S}$ with high probability. In more detail, let $Y$ be the random variable indicating the number of failed executions of \QStatePrep. Then,
\[
 \Ex[Y] \leq 10t\cdot\frac{\eps^2}{4} \qquad\text{and so}\qquad \Pr[Y > 9t] \leq \frac{5\eps^2}{18} \ll 0.01.
\]
Because each call to \QStatePrep\ makes one call to $U$, the total query complexity of \Cref{alg:q-learner} is $O\pbra{\frac{k \log k}{\eps^2} + \frac{4^k}{\eps^2}}$, completing the proof.
\end{proof}

We now prove the following lemma that we used in the proof of \Cref{thm:learning}.

\begin{lemma}\label{lem:closeness-of-phi-U}
Let $V$ denote the unitary whose Choi-Jamiolkowski isomorphism is given by $\ket{\psi_S}$, as obtained from the call to \QStatePrep\ in Step~2 of \Cref{alg:q-learner}. Then
\[\dist\pbra{U, V \otimes I^{\otimes (n - k)}} \leq \frac{\eps}{2}.\]
\end{lemma}

\begin{proof}
Since $U$ is a $k$-junta, let $R\subset [n]$ be the set of $k$ relevant variables. Let $S \subset R$ be the set of qubits with nonnegligible influence outputted by \PauliSample\ in Step 1 of \Cref{alg:q-learner}. 

Let $U = U_R \otimes I_{\bar{R}}$, where $U_R$ is a $k$-qubit unitary acting only on the relevant qubits in $R$. It is sufficient to show that $\dist(U_R, V) \leq \frac{\eps}{2}$. First, note that

\begin{align*}
\ket{v(U)} &= \ket{v(U_R)} \ket{v(I^{\otimes (n - k)})}
\end{align*}

Thus, when we measure qubits in $\bar{R}$ and $\{ \tilde{\ell} : \ell \in \bar{R}\}$, we always obtain $\ket{v(I^{\otimes (n - k)})}$, as $U$ acts trivially on qubits outside of $R$. 

Now we will consider what happens when we measure qubits in $R - S$. We will use the following decomposition of $\ket{v(U_R)}$.
\begin{align}
\ket{v(U_R)} &= \sum_{x \in Z_4^k}\hat{U_R}(x) \ket{v(\sigma_x)}\\
 &= \sum_{x : \supp(x) \cap \bar{S} = \emptyset} \hat{U}(x) \ket{v(\sigma_x)} + \sum_{x : \supp(x) \cap (R - S) \neq \emptyset} \hat{U}(x) \ket{v(\sigma_x)}\label{eq:phi-U-decomposition}
\end{align}
By \Cref{lem:pauli-sampling}, $S$ will contain all the qubits with influence larger than $ \frac{\eps^2}{4k}$ with high probability. Further, each qubit in $S$ has nonzero influence. This implies that with high probabilty, 

\begin{align*}
\sum_{x : \supp(x) \cap (R - S) \neq \emptyset} |\hat{U}(x)|^2 = \Inf_{\bar{S}}[U]
\leq \sum_{i \in \bar{S}} \Inf_{i}[U]
\leq k \cdot \frac{\eps^2}{4k} 
\end{align*}

By the decomposition in \Cref{eq:phi-U-decomposition}, measuring qubits in $(R - S) \cup \{ \tilde{\ell} : \ell \in R - S\}$ yields the state $\ket{v(I^{\otimes (|R| - |S|)})}$ with probability at least $1 - \frac{\eps^2}{4}$. Conditioned on this event, the $2k-$qubit post measurement state is as follows:

\begin{align*}
    \ket{\psi_S} &= \frac{1}{\sqrt{1 - \Inf_{R - S}[U]}} \sum_{x : \supp(x) \cap \bar{S} = \emptyset} \hat{U}(x) \ket{v(\sigma_x)} \otimes \ket{EPR}^{\otimes (k - |S|)}
\end{align*}

Let $\alpha := \frac{1}{\sqrt{1 - \Inf_{R - S}[U]}} $. Note that $1 \leq \alpha \leq \frac{1}{\sqrt{1 - \frac{\eps^2}{4}}}$. Then,

\begin{align*}
    2\ \dist^2 (V, U_R) &= ||\ket{\psi_S} - \ket{v(U_R)}||^2\\
    &= || \sum_{x : \supp(x) \cap \bar{S} = \emptyset} \hat{U}(x) \ket{v(\sigma_x)} + \sum_{x : \supp(x) \cap (R - S) \neq \emptyset} \hat{U}(x) \ket{v(\sigma_x)}\\ 
    & - \alpha \sum_{x : \supp(x) \cap \bar{S} = \emptyset} \hat{U}(x) \ket{v(\sigma_x)}||^2\\
    &= (\alpha - 1)^2  \sum_{x : \supp(x) \cap \bar{S} = \emptyset} |\hat{U}(x)|^2 + \sum_{x : \supp(x) \cap (R - S) \neq \emptyset} |\hat{U}(x)|^2\\
    &\leq (\frac{1}{\sqrt{1 - \frac{\eps^2}{4}}} - 1)^2 + \Inf_{R - S}[U]\\
    &\leq 2 \frac{\eps^2}{4}
\end{align*}

This shows that $ \dist (V \otimes I^{\otimes (n - k)}, U) = \dist (V, U_R) \leq \eps/2$.
\end{proof}

The following lemma is analogous to Lemma IV.4~in~\cite{Atici2007}.

\begin{lemma}\label{lem:pauli-sampling}
Let $U\in\calU_N$ be a unitary acting non-trivially on qubits in $R \subset [n]$. Then \PauliSample$(U,\eps,|R|)$ makes $t = O\pbra{\frac{\log |R|}{\eps}}$ membership queries to $U$ and outputs with high probability a list $S \subset [n]$ that satisfies the following properties:
\begin{enumerate}
    \item $S$ contains all qubits $i\in[n]$ such that $\Inf_i [U] \geq \eps$; and 
    \item All qubits $i$ in $S$ have nonzero influence, i.e. $\Inf_i [U] > 0$.
\end{enumerate}
\end{lemma}

\begin{proof}
 If $\Inf_i [U] \geq \eps$, then the probability $i$ does not occur in S is at most $(1 - \eps)^{t} \leq \frac{1}{100|R|}$. By the union bound, $S$ will contain every $i$ such that $\Inf_i [U] \geq \eps$ with probability at least 99/100. The second item follows from the fact that if $i\in[n]$ is Pauli-sampled, there must exist $x\in\Z_4^n, i \in \supp(x)$ such that $\wh{U}(x) \neq 0$.
\end{proof}

\subsection{Learning Lower Bound}
\label{subsec:learn-lb}

Finally, we present a nearly-matching lower bound for the query complexity of learning quantum juntas. Although it is commonly stated that process tomography requires $\Omega(4^n)$ queries, we have not been able to identify in the literature a formal lower bound proof. Thus, we provide the following proof for completeness.

\begin{theorem}\label{thm:learning-lb}
Any algorithm for learning quantum $k$-juntas with error $\eps$ requires $\Omega(4^k \log (1/\eps)/k)$ queries.
\end{theorem}
\begin{proof}
We prove this lower bound via a communication complexity argument. In particular, we reduce the \textsc{Input Guessing} game to learning quantum juntas. The \textsc{Input Guessing} game with domain size $K$ is a two-party communication task where one party (named Alice) receives an uniformly random input $x$ from $\{1,2,\ldots,K\}$ and the other party (named Bob) has to output a guess for $x$ after engaging in two-way communicating with Alice. We consider the model of \emph{quantum communication}, where Alice and Bob can exchange qubits with each other. A classic result of Nayak~\cite{nayak1999optimal} implies the following lower bound on the communication complexity of the \textsc{Input Guessing} game. % on the parties' maximum success probability:

\begin{theorem}[Lower bound for \textsc{Input Guessing} game~\cite{nayak1999optimal}]
\label{thm:nayak}
Any quantum communication protocol that solves the \textsc{Input Guessing} game with domain size $K$ and success probability $p$ requires exchanging $\log K - \log \frac{1}{p}$ qubits between the parties.
\end{theorem}

Let $\calA$ denote an algorithm that learns quantum $k$-juntas, assuming it is told which $k$ of the $n$ qubits are relevant. Note that the problem of learning quantum $k$-juntas without this additional information is at least as hard. Without loss of generality, assume the first $k$ qubits are relevant. 

Suppose $\calA$ makes $q$ queries, achieves error $\eps$ and achieves constant success probability. Then we construct a quantum communication protocol for the \textsc{Input Guessing} game with domain size $K = \Omega \Big( (1/\eps)^{4^k} \Big)$, communication complexity $O(kq)$, and constant success probability. By \Cref{thm:nayak}, this implies that
\[
    kq \geq \Omega(\log K) = \Omega(4^k \log(1/\eps))
\]
which implies the desired lower bound.

Let $K$ denote the size of a maximal $\eps$-packing of the space of $k$-qubit unitary matrices, with respect to the distance measure $\dist(\cdot,\cdot)$. In other words, this is the maximal number of disjoint $\eps$-balls in the space of $k$-qubit unitaries. By standard volume arguments (see ~\cite{szarek1997metric}), since $\dist(\cdot,\cdot)$ is a unitarily invariant distance measure, 
% \henry{double check here} 
$K$ is at least $\Omega \Big( (1/\eps)^{4^k} \Big)$. Let $\{U_1,\ldots,U_K\}$ denote an enumeration of the maximal $\eps$-packing.

Suppose Alice gets a random input $x \in \{1,2,\ldots,K\}$. Bob will simulate the algorithm $\cal{A}$. Whenever $\cal{A}$ has to make a query to the oracle, Bob sends the first $k$ qubits of his query register to Alice, then Alice applies the $k$-qubit unitary $U_x$ to the register, and then sends the register back. Indeed, if $\calA$'s query register is on $n$ qubits total, then the effective $n$-qubit unitary applied in this simulated execution of $\calA$ is $U = U_x \otimes I$ where $U_x$ acts on the first $k$ qubits and $I$ acts on the remaining $n - k$ qubits. Bob continues in this fashion until the algorithm $\cal{A}$ terminates and outputs (with constant probability) a classical description of a unitary $V$ such that $\dist(U_x \otimes I, V) \leq \eps$ where, by the correctness of the algorithm $\calA$, $V$ is a quantum $k$-junta $V' \otimes I$ that acts trivially on all qubits except the first $k$. Thus we have that $\dist(U_x,V') \leq \eps$, and by definition of an $\eps$-packing, $U_x$ is the unique member of the packing that has distance $\eps$ to $V'$. Thus Bob can uniquely identify Alice's input $x$ with constant probability. The total communication complexity of this protocol is $2kq$. 
\end{proof}

%\cite{Atici2007} showed the following:

% \begin{lemma} (\cite{Atici2007} Observation IV.8)
% Any 1/10-learning quantum membership query algorithm for boolean $k$-juntas that uses only $\frac{1}{101} 2^k$ classical membership queries must additionally use $\Omega(2^k)$ quantum membership queries.
% \end{lemma}

% Certainly, any quantum algorithm using no classical membership queries will need $\Omega(2^k)$ quantum membership queries to learn a classical boolean $k-$junta. By encoding a boolean $k-$junta as a diagonal unitary (the same reduction as the testing lower-bound), it's clear that learning a quantum $k-$junta must also require at least $\Omega(2^k)$ quantum membership queries.

\section*{Acknowledgements}

S.N. is supported by NSF grants CCF-1563155 and by CCF-1763970. H.Y. is supported by AFOSR award
FA9550-21-1-0040, NSF CAREER award CCF-2144219, and the Sloan Foundation. The authors would like to thank Rocco A. Servedio and Xi Chen for helpful discussions. We would also like to thank Zongbo Bao for pointing out an error in \Cref{alg:raw-IT} in an earlier version of this paper, as well as Vishnu Iyer and Michael Whitmeyer for helpful comments. 

\bibliography{references}{}
\bibliographystyle{alpha}

\appendix
%\section{A $\Omega(\sqrt{k})$ Lower Bound for Pauli Sampling Oracles}
%\label{sec:pauli-lb}
%
%\red{
%
%The high-level idea is to essentially reduce to the classical case: We have a reasonably nice relationship between the Fourier coefficients and the Pauli coefficients when you encode a function $f\isazofunc$ as a unitary $\calO_f$, and the idea is to reduce from the ``hard'' instances from the Atici-Servedio paper and encode them as unitaries.
%
%}
\section{An Alternative Characterization of Influence}
\label{appendix:influence}

Below is an alternative characterization of the influence of a set of variables on a unitary. Our learning and testing algorithms do not make use of this characterization, but it may be of independent interest. 

\begin{lemma} (Equivalent characterization of influence)
\label{lemma:useful-inf-def}
Given $U\in\calU_N$ and $j\in[n]$, we have 
\[\Inf_j[U] = 1 - \frac{1}{2^{n + 1}} \Tr\pbra{(\Tr_j U^\dagger) (\Tr_j U)}.\]
More generally, for $S\sse[n]$, we have 
\[\Inf_S[U] = 1 - \frac{1}{2^{n + |S|}} \Tr\pbra{(\Tr_S U^\dagger) (\Tr_S U)}.\]
\end{lemma}

\begin{proof}

Note that for $S\sse[n]$ and $x\in\Z_4^n$, we have that 
\[\Tr_S(\sigma_x) = 0 \qquad\text{if and only if}\qquad S\cap \supp(x) \neq \emptyset.\]
This is immediate from the fact that the only Pauli matrix with non-zero trace is $\sigma_0 = I$. Writing $U$ in the Pauli basis, $\Tr_S(U)$ has the following form
\begin{equation}
	\Tr_S(U) = \Tr_S\pbra{\sum_{x\in\Z_4^n}\wh{U}(x)\sigma_x} = 2^{|S|}\sum_{x : \supp(x)\cap S = \emptyset} \wh{U}(x)\sigma_x.\label{eq:partial-trace-U}
\end{equation}
Using this characterization, it follows that
\begin{align}
\frac{1}{2^{n - |S|}} \Tr\pbra{\Tr_S(U)^\dagger \Tr_S (U)} &= \frac{1}{2^{n - |S|}} \abra{ \Tr_S(U), \Tr_S(U)} \nonumber\\
&= 2^{2|S|} \sum_{x : \supp(x)\cap S = \emptyset} |\wh{U}(x)|^2 \nonumber\\
&= 2^{2|S|} \pbra{1 - \Inf_S[U]} \nonumber
\end{align}
where the second equality follows from Parseval's formula and equation \ref{eq:partial-trace-U}, while the final equality is because $\|U\|^2=N$ for all unitaries $U\in\calU_N$. The lemma follows by rearranging the final expression above.
\end{proof}

\end{document}